\documentclass[11pt,times]{article}

\usepackage[page-ref=none]{acro}

\usepackage{algorithm}
\usepackage{algorithmic}

\usepackage[sumlimits]{amsmath}
\usepackage{amssymb}
\usepackage{amsthm}

\usepackage{booktabs}

\usepackage{color}
\usepackage{enumitem}

\usepackage[mathscr]{eucal}
\usepackage{fancybox}
\usepackage{float}

\usepackage[text={6.5in,9.5in},centering]{geometry}

\usepackage{graphics}
\usepackage{graphicx}

\usepackage[hidelinks]{hyperref}

\usepackage{ifpdf}
\usepackage{ifthen}
\usepackage{lineno}
\usepackage{listings}
\usepackage{makecell}

\usepackage{multicol}
\usepackage{multirow}

\usepackage{pgfplots}

\usepackage{stackengine}

\usepackage{stmaryrd}
\usepackage{subfig}

\usepackage{tikz}
\usetikzlibrary{arrows}
\usetikzlibrary{fadings}

\usepackage{tikz-cd}

\usepackage{txfonts}    

\usepackage{todonotes}

\usepackage{url}
\usepackage{varioref}
\usepackage{xcolor}

\DeclareAcronym{AMP}{
  short = AMP ,
  long  = asymmetric multicore processor
}

\DeclareAcronym{BCDomain}{
  short = bc-domain ,
  long  = bounded-complete domain
}

\DeclareAcronym{BVP}{
  short = BVP ,
  long  = boundary value problem
}

\DeclareAcronym{CLEVER}{
  short = \textsc{Clever} ,
  long  = Cross Lipschitz Extreme Value for nEtwork Robustness
}

\DeclareAcronym{CPO}{
  short = cpo ,
  long  = complete partial order
}

\DeclareAcronym{CPPO}{
  short = cppo ,
  long  = complete pointed partial order
}

\DeclareAcronym{CPS}{
  short = CPS ,
  long  = cyber-physical system
}

\DeclareAcronym{DAG}{
  short = DAG ,
  long  = directed acyclic graph
}

\DeclareAcronym{DBF}{
  short = DBF ,
  long  = Digital Beam Former
}

\DeclareAcronym{DCPO}{
  short = dcpo ,
  long  = directed-complete partial order
}

\DeclareAcronym{DSL}{
  short = DSL ,
  long  = domain-specific language
}

\DeclareAcronym{EDSL}{
  short = EDSL ,
  long  = embedded domain-specific language
}

\DeclareAcronym{FBP}{
  short = FBP ,
  long  = free boundary problem
}

\DeclareAcronym{FivePs}{
  short = 5Ps ,
  long  = five Ps
}

\DeclareAcronym{glb}{
  short = glb ,
  long  = greatest lower bound
}

\DeclareAcronym{HS}{
  short = HS ,
  long  = hybrid system
}

\DeclareAcronym{ISA}{
  short = ISA ,
  long  = instruction set architecture
}

\DeclareAcronym{IVP}{
  short = IVP ,
  long  = initial value problem
}

\DeclareAcronym{lub}{
  short = lub ,
  long  = least upper bound
}

\DeclareAcronym{ODE}{
  short = ODE ,
  long  = ordinary differential equation
}

\DeclareAcronym{PC}{
  short = PC ,
  long  = Pulse Compression
}

\DeclareAcronym{PCROP}{
  short = PC-ROP ,
  long  = PDE-constrained rearrangement optimization problem
}

\DeclareAcronym{PDCPO}{
  short = pointed dcpo ,
  long  = pointed directed-complete partial order
}

\DeclareAcronym{PDE}{
  short = PDE ,
  long  = partial differential equation
}

\DeclareAcronym{PI}{
  short = PI ,
  long  = principal investigator
}

\DeclareAcronym{POSET}{
  short = poset ,
  long  = partially ordered set
}

\DeclareAcronym{RE}{
  short = r.~e. ,
  long  = recursively enumerable
}

\DeclareAcronym{RNN}{
  short = RNN ,
  long  = recurrent neural network
}

\DeclareAcronym{ROP}{
  short = ROP ,
  long  = rearrangement optimization problem
}

\DeclareAcronym{SAC}{
  short = \textsc{Sac} ,
  long  = Single Assignment C
}

\DeclareAcronym{SBSE}{
  short = SBSE ,
  long  = search-based software engineerign
}

\DeclareAcronym{SOP}{
  short = SOP ,
  long  = shape optimization problem
}

\DeclareAcronym{TTE}{
  short = TTE ,
  long  = Type-II Theory of Effectivity
}

\newcommand{\dperp}{\mathrel{\bot\!\!\!\bot}}
\newcommand{\interval}{\mathbb{I}}
\newcommand{\nat}{\mathbb{N}}

\newcommand{\cpo}[1][s]{\ifthenelse{\equal{#1}{s}}{complete partial
order}{Complete partial order}}

\newcommand{\fpoint}[1][s]{\ifthenelse{\equal{#1}{c}}{Floating-point}{floating-point}}

\newcommand{\kc}[1][s]{\ifthenelse{\equal{#1}{s}}{Kolmogorov
complexity}{Kolmogorov complexities}} 

\newcommand{\ptime}[1][s]{\ifthenelse{\equal{#1}{c}}{Polynomial-time}{polynomial-time}}
\newcommand{\pspace}[1][s]{\ifthenelse{\equal{#1}{c}}{Polynomial-space}{polynomial-space}}

\newcommand{\NFL}[1][c]{\ifthenelse{\equal{#1}{s}}{no free lunch}{No
Free Lunch}}

\newcommand{\sequence}[5][i]{\ifthenelse{\equal{#1}{i}}{\ensuremath{\left< {#2}_{#3}, {#2}_{#4},
\ldots,{#2}_{#5}, \ldots \right>}}{\ensuremath{\left< {#2}_{#3}, {#2}_{#4},
\ldots,{#2}_{#5} \right>}}}

\theoremstyle{plain} 

\newtheorem{definition}{Definition}[section]

\newtheorem{corollary}[definition]{Corollary}
\newtheorem{example}[definition]{Example}
\newtheorem{lemma}[definition]{Lemma}

\newtheorem{proposition}[definition]{Proposition}

\newtheorem{theorem}[definition]{Theorem}

\newcounter{defenumalph}

\newcounter{defenum}

\newcounter{saveeqn}

\newcommand{\arrayoptions}[2]{\setlength{\arraycolsep}{#1}\renewcommand{\arraystretch}{#2}}

\DeclareFontFamily{OT1}{pzc}{}
\DeclareFontShape{OT1}{pzc}{m}{it}{<->s*[1.30]pzcmi7t}{}
\DeclareMathAlphabet{\mathpzc}{OT1}{pzc}{m}{it}

\newcounter{saveenum} 

\title{\textbf{A Domain-Theoretic Foundation for Imprecise Probability and Credal Sets}}

\author{Abbas Edalat\thanks{Department of Computing, Imperial College
    London, London, United Kingdom, Email:
    \href{mailto:a.edalat@imperial.ac.uk}
    {\texttt{a.edalat@imperial.ac.uk}}} \and
  Pietro Di Gianantonio\thanks{Department of Mathematics and
    Informatics, University of Udine, Udine, Italy,
    Email:\href{mailto:pietro.digianantonio@uniud.it}
    {\texttt{pietro.digianantonio@uniud.it}}} \and
    Amin Farjudian\thanks{School of Mathematics, University of
    Birmingham, Birmingham, United Kingdom, Email:
    \href{mailto:A.Farjudian@bham.ac.uk}
    {\texttt{A.Farjudian@bham.ac.uk}}}}

\date{}

\begin{document}

\maketitle

\begin{abstract}
  We develop a domain-theoretic framework for imprecise probability
  reasoning and inference on general topological spaces with a
  countably based continuous lattice of open sets. We address two
  distinct forms of uncertainty: partial or incomplete event
  descriptions, and sets of probability distributions as represented
  by credal sets---as well as their combination. Within this
  framework, we construct a theory of conditional probability and
  derive novel inference rules for performing Bayesian updating in the
  presence of these two complementary types of imprecision. These
  results are extended to a theory of conditional independence for
  imprecise probabilistic events. We also formulate logical predicates for conditional probability, Bayesian updating, and conditional independence, and we obtain the relevant soundness and completeness results. A key contribution is the
  construction of a Scott-continuous mapping from any credal set to
  the domain of intervals, providing a domain-theoretic realisation of
  classical results from capacity theory and Choquet
  integration. Finally, we introduce and study a new family of credal
  sets generated by iterated function systems with imprecise
  probability weights, broadening the scope of computationally
  tractable imprecise probabilistic models. The resulting computable
  framework unifies logical, topological, and measure-theoretic
  perspectives on uncertainty, supporting robust probabilistic
  inference under partial and set-valued information.
\end{abstract}


\section{Introduction}

Imprecise probability provides a robust framework for reasoning under uncertainty when information is partial, ambiguous, or set-valued. It generalises classical probability by allowing sets of distributions (credal sets) and interval-valued probabilities, enabling more cautious inferences in safety-critical applications.

We consider second-countable locally compact sober topological spaces.  We refer to such a space as a {\em basic} topological space. In these spaces the lattice of open sets is a countably based continuous lattice~\cite{gierz2003continuous}, representing a spatial locale with a countable basis. These basic spaces include separable locally compact metric spaces as well as countably based continuous domains. Moreover, any Polish space is the space of maximal elements of its continuous domain of formal balls~\cite{edalat1998computational}. This means that basic topological spaces cover all the standard spaces used in probability theory.

In addition, any continuous probability valuation on such a space extends to a Borel measure\cite{alvarez2000extension,Keimel_Lawson}. For Hausdorff spaces, the extension is unique by outer regularity of the resulting Borel measure. 

Let $D$ be a basic space with its lattice of open sets $\mathcal{O}_D$. We consider a pair $(O_1,O_2)\in \mathcal{O}_D\times \mathcal{O}_D$ of disjoint open sets of $D$ as an observational event where $O_2$ is considered as the current information about the exterior of $O_1$. We define the {\em domain of events} of $D$---denoted by $\mathbb{E}(D)$---be the set of disjoint pairs of open sets ordered componentwise by subset inclusion. This gives a bounded complete domain that had been considered as a computable framework for solid modelling~\cite{edalat2002foundation}.

On disjoint events, it is possible also to define an interval notion of probability of events. 
Let $ \interval[0,1]$ be the domain of non-empty closed intervals of the unit interval ordered by reverse inclusion. 

\begin{definition}\label{int-prob)}
 Given a probability distribution $\sigma$ on $D$, the {\em interval probability map} on the space of disjoint event-pair is given by $\operatorname{Pr}_\sigma: \mathbb{E}( D) \to \interval[0,1]$ with:
 \begin{equation}\label{defining-int-prob}
 \operatorname{Pr}_\sigma((O_1,O_2))=[\sigma(O_1),1-\sigma(O_2)].   
 \end{equation} 
\end{definition}

It is easily seen that $\operatorname{Pr}_\sigma$ is Scott continuous. Given an event-pair $(O_1,O_2)$, the real interval $\operatorname{Pr}_\sigma((O_1,O_2))\in \interval[0,1]$ captures the minimum and maximum probability $\sigma(O_1)$ as the event-pair is refined in the information order.

In the following sections, we develop a domain-theoretic foundation for conditional probability, Bayesian updating, and conditional independence in this imprecise setting. A key result is the Scott-continuous envelope map from credal sets to interval probabilities, which lifts capacity-theoretic ideas and Choquet integration into a domain framework~\cite{Choquet1954,grabisch2016set,Augustin2014,Goubault2013}. We also introduce a new family of credal sets from iterated function systems with imprecise weights.

While interval arithmetic has been applied to Bayes' rule in engineering contexts \cite{ferson2003experimental}, 
and robust Bayesian analysis considers sets of priors \cite{berger1985statistical}, 
to our knowledge the \emph{monotonicity-based derivation} of exact endpoint formulas 
has not been previously published. We prove (Lemma~\ref{interval-bayes}) that the classical Bayesian update 
\[
f_b(P(H), P(E|H), P(E|\neg H)\bigr) = \frac{P(H)P(E|H)}{P(H)P(E|H) + P(E|\neg H)(1-P(H))}
\]
is monotone increasing in $P(H)$ and $P(E|H)$ and decreasing in $P(E|\neg H)$ (where $P$, $H$ and $E$ stand for Probability, Hypothesis and Evidence respectively). This yields the sharp interval extension formula in Corollary~\ref{def-int-bayes-rule} and forms the foundation for our 
 inference rules (B1)--(B4) and their completeness proof (Theorem~\ref{bayes-complete}). 

 Since the lattice of open sets of a basic space is countably based and continuous, it can be given an effective structure so that computable opens and computable functions on the lattice can be enumerated; see~\cite{plotkin1981post,smyth1977effectively}. This leads to a computable framework for imprecise probabilities and credal sets. 
\subsubsection*{Notational Convention} We use $D$ for any basic space, Hausdorff or domain. When we specifically, deal only with a Hausdorff basic space, we denote it by $X$ rather than $D$.

\section{Domain theoretic basics}

Recall that a continuous probability valuation $\sigma$ on a complete
lattice $L$ is a Scott continuous map $\sigma:L\to [0,1]$, with the
modularity property
\[\sigma(a)+\sigma(b)=\sigma(a\vee b)+\sigma(a\wedge b).\]
We consider probability continuous valuations $\sigma$ on the lattice
of open subsets of a basic space $D$. The weak topology on $P(D)$ is
the weakest topology for which the evaluation maps
$\operatorname{ev}_O:P(D)\to ([0,1],\leq)$ are continuous, in which
$O\in\mathcal{O}_D$ ranges over open sets and
$\operatorname{ev}_O(\sigma)=\sigma(O)$. Thus, $\operatorname{ev}_O$
is a lower semi-continuous map, i.e., a Scott continuous map into
$([0,1],\leq)$. If $D$ is a continuous domain then the weak topology
on $P(D)$ coincides with the Scott topology on $P(D)$.

 The Scott topology on $L$ is the natural topology for computation. In Appendix~\ref{appendix:not_Lawson_cont}, we will show that a continuous valuation is not in general continuous with respect to the Lawson topology, which is finer than the Scott topology. 
 
\subsection{Approximable relations}

Following the established framework and tradition for imprecise
probabilities, as extensively described in the classical pioneering
work~\cite{Walley1991}, we will formulate predicates for the basic
domain-theoretic constructions. This can be achieved by the rich
theory of approximable mappings, locales, and Stone duality in domain
theory~\cite{Scott70,smyth1977effectively,abramsky1995domain,abramsky1991domain,vickers1989topology}.

In fact, we will formulate a predicate $G$ that can be used (i) to define a basis of Scott open subsets of $P(D)$, for a basic space $D$ and, (ii) to characterise a given continuous valuation on $\mathcal{O}(D)$. We consider a countable basis $B_D$ of the lattice of open subsets of a basic space $D$, closed under finite union and finite intersection.

We define a basis of open sets on the spaces of continuous valuations by sets given as follows:
\[
G(v;p) = \{ \sigma \in P(D) \mid p < \sigma(v) \}
\]
Let $\mathbb{Q}_0:=\mathbb{Q}\cap(0,1)$. Each continuous valuation $\sigma$ can be fully described by the set of all open sets $G(v;p)$ to which it belongs. In the spirit of approximable mappings, a continuous probability valuation $\sigma$ is characterised by the relation $G_\sigma\subset B_D\times \mathbb{Q}_0$ containing all the pair $(v, p)$ such that $p<\sigma(v)$. 

Formally, we define an approximable mapping over an abstract basis, a set with a transitive, interpolative relation~\cite{abramsky1995domain} for a given continuous probability valuation $\sigma$ by the relation $G\subset B_D\times \mathbb{Q}_0$ satisfying the following axioms.  
The first four are conditions to represent a Scott continuous map and the last two capture the modularity conditions:

\begin{enumerate}
    \item $ \ G(D; p)$
    \item $ \ G(\emptyset; p) \implies \operatorname{False}$

    \item $
    \ v \subseteq v' \ \wedge \ p \geq p' \wedge G(v;p) \implies G(v';p')$
    \item $
    G(v;p) \implies \exists \ v' \ll v, \, \ p' > p . \, G(v';p')$
    \item $  (p+q = p'+q') \ \wedge \ G(u;p) \ \wedge \  G(v;q) \implies G((u \cup v);p') \vee G((u \cap v);q')$ (Sub-modularity)
    \item  $(p+q = p'+q') \ \wedge G((u \cup v);p) \ \wedge \  G((u \cap v);q) \implies G(u;p') \vee G(v;q')$ (Super-modularity)
\end{enumerate}
These axioms follow the same schema as those in \cite{MoshierJung2002}.


\section{Event domain of basic spaces}

Given a basic space $D$, we consider its open sets as observable or semi-decidable predicate~\cite{abramsky1991domain,smyth1977effectively}. Since taking complements of an event in probability and statistics is an essential tool and since the complement of an open set may not be open, we consider approximation of the exterior of an open set by disjoint opens. This leads us to 
 the domain $\mathbb{E}(D)$ of {\em events} of $D$ as the partial order of disjoint pairs of open sets ordered component-wise by subset inclusion:
\[\mathbb{E}(D)=\left(\{(O_1,O_2): O_1,O_2\in \mathcal{O}_D, O_1\cap O_2=\emptyset\},\subseteqq\right)\]
In fact, this domain is implicitly used in~\cite[Section~9.1]{GoubaultJT23}. 
 \begin{proposition}~\cite{edalat2002foundation}
    For any basic space $D$, the partial order $\mathbb{E}(D)$ is a countably based bounded complete domain.
  \end{proposition}

  We can define various operations on events as follows.
  The two binary operations of intersection $\cap:\mathbb{E}(D)\times
  \mathbb{E}(D)\to \mathbb{E}(D)$ and union $\cup:\mathbb{E}(D)\times \mathbb{E}(D)\to \mathbb{E}(D)$ of events are defined for a basic space $D$ as follows:
\begin{enumerate}
    \item {\bf Intersection of events:}  $(O_1,O_2)\cap(U_1,U_2)=(O_1\cap U_1, O_2\cup U_2)$.
    \item {\bf Union of events:} $(O_1,O_2)\cup (U_1,U_2)=(O_1\cup U_1, O_2\cap U_2)$.

\setcounter{saveenum}{\value{enumi}}

\end{enumerate}

The binary operations $\cap$ and $\cup$ are Scott continuous~\cite{edalat2002foundation}.The basis $B_D$ of $D$ induces a basis $B_{\mathbb{E}(D)}$ closed under finite intersection and finite union.

 Let $D_1$ and $D_2$ be basic spaces; we define the product of their event spaces:
 
 \[\otimes:\mathbb{E}(D_1)\times \mathbb{E}(D_2)\to \mathbb{E}(D_1\times D_2)\]
with $(O,U)\mapsto (O_1\times U_1,(O_2\times D_2)\cup (D_1\times U_2) )$.
Since the Cartesian product and finite union are continuous operations, so is $\otimes$.

More generally, finite products are defined by associativity as follows. If     $U^i=(U^i_1,U^i_2)\in \mathbb{E}(D_i)$, for $i\in I$, is a a finite set of events in the domains $D_i$ for $i\in I$, respectively, their product in the product space $\prod_{i\in I}D_i$ is given by:
\begin{enumerate}
\setcounter{enumi}{\value{saveenum}}
    \item {\bf Finite product of events:} \[\otimes_{i\in I}U^i=\left(\prod_{i\in I}U^i_1,\bigcup_{i\in I} \left( U^i_2\times \prod_{j\neq i} D_{j}\right)\right).\]
\end{enumerate}

Let $\sigma$ be a continuous probability valuation on $D$ and
$(V_1,V_2)$ be an event in $D$, i.e., a pair of disjoint basic open
sets in $B$. We can define the map
\begin{equation}
\arrayoptions{0.5ex}{1.2}
\begin{array}{cl}
  \mathcal{E}: & P(D)\times \mathbb{E}(D)\to \mathbb{I}[0,1] \\
  & (\sigma,(V_1,V_2))\mapsto \sigma(V_1,V_2)=\left[\sigma(V_1),1-\sigma(V_2)\right]
\end{array}
\end{equation}
 Note that $\mathcal{E}$ is well-defined:  the condition $V_1\cap V_2=\emptyset$ implies $\sigma(V_1)+\sigma(V_2)\leq 1$ and thus $\sigma(V_1)\leq 1-\sigma(V_2)$. If $X$ is a Hausdorff space then $P(X)$ is equipped with its weak topology, whereas if $D$ is a continuous domain $P(D)$ is equipped with its Scott topology. In either case, the map $\mathcal{E}$ is a continuous map: If $X$ is Hausdorff, then the map $\operatorname{ev}_O:P(X)\to [0,1]$ is lower semi-continuous for any open set $O\subset D$, which implies the continuity of $\mathcal{E}$. If $D$ is a continuous domain then  the map $\operatorname{ev}_O:P(D)\to [0,1]$ is Scott continuous, which implies the continuity of $\mathcal{E}$.

 For a given continuous probability valuation $\sigma$, we write
 $\mathcal{E}_\sigma: \mathbb{E}(D)\to \mathbb{I}[0,1]$ with
 $\mathcal{E}_\sigma(O_1,O_2)=\mathcal{E}(\sigma,(O_1,O_2))$.  The map
 $\mathcal{E}_\sigma$ is then represented by the approximable map:
 $\widehat{G}\subseteq B_{\mathbb{E}(D)}\times \mathbb{I}\mathbb{Q}_0$
 defined by
 \[\widehat{G}((u,v),[p,q])\ \mbox{ if }\ (\sigma(u)>p) \wedge (
   1-\sigma(v)<q ).\]Since $1-\sigma(v)<q $ iff $\sigma(v)>1-q$ we
 obtain:
\[\widehat{G}((u,v),[p,q]) \iff (\sigma(u)>p) \wedge ( 1-\sigma(v)<q )\]\[\iff (\sigma(u)>p) \wedge  (\sigma(v)>1-q )\iff G(u; p)\,\wedge G(v; (1-q))\]

\section{Credal sets}
The modern term \emph{credal set} for a convex set of probability distributions was standardised in later treatments \cite{Augustin2014}, though the underlying theory was developed by Walley \cite{Walley1991} under the name ``sets of probability measures.''

 Recall that if $X$ is a locally compact Hausdorff space, its upper space $U(X)$, the space of non-empty compact subsets of $X$ ordered by reverse inclusion, is a bounded complete domain with the topological embedding $e:X\to U(X)$ given by $e(x)=\{x\}$~\cite{edalat1995dynamical}. We identify $x\in X$ with the singleton $\{x\}\in U(X)$. 
 
 For a compact metric space $X$, we consider the domain $U(P(X))$ as the space of compact credal sets for $X$. Any $\sigma\in P(X)$ gives (via the singleton set $\{\sigma\}\in U(P(X))$) a maximal element of $U(P(X))$. Then $P(X)\subseteq U(P(X))$ is dense with respect to the Scott topology on $U(P(X))$ and thus by the property of densely injective spaces~\cite{gierz2003continuous}, we get for each $O\in\mathcal{O}_X$, an extension 
 
\[\widehat{\operatorname{ev}_O}:U(P(X))\to ([0,1],\leq)\]
with $\widehat{\operatorname{ev}_O}(\{\sigma\})=\sigma(O)$ and $\widehat{\operatorname{ev}_O}(K)=\min\{\sigma(O):\sigma\in K\}$. 

The map $\mathcal{E}:P(X)\to [ \mathbb{E}(X)\to \mathbb{I}[0,1] ]$
extends to a Scott continuous map
$\hat{\mathcal{E}}:U(P(X))\to (\mathbb{E}(X)\to \mathbb{I}[0,1])$ with
$\hat{\mathcal{E}}(\{\sigma\})=\mathcal{E}(\sigma)$, for
$\sigma\in P(X)$. Thus, we obtain an extension of $\mathcal{E}$ on the
space of credal sets given by:
\begin{equation}\label{extension-to-credal}
\hat{\mathcal E}(K)(O_1,O_2)=\left[\inf_{i\in I}\sigma_i(O_1),1-\inf_{i\in I}\sigma_i(O_2)\right].    
\end{equation}

When $(O_1,O_2)$ satisfies the classical event condition $\operatorname{CE}(O)$ (i.e., $\sigma(O_1)+\sigma(O_2)=1$ for all $\sigma$), we obtain common grounds with the theory of capacities. In  fact, the right-hand side of Equation~\ref{extension-to-credal} computes precisely the \emph{lower and upper probabilities} (or \emph{belief and plausibility}) of the event $O_1$ with respect to the credal set $K$.  Indeed,
\[
\underline{P}(O_1)=\inf_{\sigma\in K}\sigma(O_1),\]\[
\overline{P}(O_1)=\sup_{\sigma\in K}\sigma(O_1)=1-\inf_{\sigma\in K}\sigma(O_2),
\]
so that
\[
\hat{\mathcal{E}}(K)(O_1,O_2)=\bigl[\,\underline{P}(O_1),\,\overline{P}(O_1)\,\bigr].
\]
The lower and upper probabilities, as computed above domain-theoretically, correspond to the Choquet integral of indicator functions with respect to the capacities induced by the credal set $K$. 

The sub-domain of $U(P(X))$ consisting of non-empty convex compact
subsets of $P(X)$ has a basis consisting of convex polytopes in $P(X)$
as every convex subset of $P(X)$ is the intersection of convex
polytopes way-below it. For such a basis element with vertices
$\sigma_i$ with $i\in I$, we have:
$\widehat{\operatorname{ev}_O}(K)=\min\{\sigma_i(O):i\in I\}$.

Thus, if $K\in U(P(X))$ is a convex polytope in $P(X)$ with vertices $\sigma_i$ for $i\in I$, then for any $(O_1,O_2)\in \mathbb{E}(X)$ we have:

\begin{equation*}
\hat{\mathcal E}(K)(O_1,O_2)=\left[\min_{i\in I}\sigma_i(O_1),1-\min_{i\in I}\sigma_i(O_2)\right].    
\end{equation*}

\section{Conditional probability of events}
\label{cond-prob}

We define the conditional probability map with respect to the event domain $\mathbb{E}(D)$ as the map $\mathcal{C}:P(D)\times \big({\mathbb{E}(D)}\times {\mathbb{E}(D)} \big)\to \mathbb{I}[0,1]$:
\begin{equation}\label{cond-prob-map-events}\mathcal{C}(\sigma,((V_1,V_2),(O_1,O_2)) = \left[\frac{\sigma(V_1\cap O_1)}{1-\sigma(O_2)},1-\frac{\sigma(V_2\cap O_1)}{1-\sigma(O_2)}\right],\end{equation}
provided $\sigma(O_1)>0$, which implies $1-\sigma(O_2)>0$. 
To this effect, for $O = (O_1, O_2)\in \mathbb{E}(D)$, we define the positivity predicate:
  \[ \operatorname{Pos}(O) = \exists p>0: G(O_1;p)\]

  If $\sigma(O_1)+\sigma(O_2)=1$, then the left end point of the
  interval in Equation~\ref{cond-prob-map-events} is reduced to the
  classical conditional probability of $V_1$ given $O_1$. If, in
  addition, $\sigma(V_1)+\sigma(V_2)=1$, then the right end point also
  equals the left end point, i.e., we simply have the point-valued
  classical conditional probability.
  
  The map $\mathcal{C}$ is easily checked (like the map $\mathcal{E}$) to be continuous for any basic space $D$, be it Hausdorff or non-Hausdorff.

  As for $\mathbb{E}$, if $X$ is a compact metric space, then we have the Scott continuous extension of $\mathcal{C}$ to the space of compact credal subsets of $P(X)$ of type:
  \[\hat{\mathcal{C}}:U(P(X))\times {\mathbb{E}(X)}\times {\mathbb{E}(X)}\to \mathbb{I}[0,1].\]

  For a convex polytope $K\in U(P(X))$, with vertices $\sigma_i$ for $i\in I$, we also have:

\[\hat{\mathcal{C}}(K,((V_1,V_2),(O_1,O_2)) = \left[\min_{i\in I}\frac{\sigma_i(V_1\cap O_1)}{1-\sigma_i(O_2)},1-\min_{i\in I}\frac{\sigma_i(V_2\cap O_1)}{1-\sigma_i(O_2)}\right]\]

\subsection{Conditional probability predicates}
\label{subsec:cond-prob-pred}

Since the Scott continuous map $\mathcal{C}$ is given by a pair of rational functions of the input continuous valuation $\sigma$ on the input opens or their intersections, one can in principle obtain the approximable mapping representing $\mathcal{C}$ by taking compositions of the given operations. This method, however, leads to quite complicated expressions. A more natural and straightforward technique is to formulate two key predicates for the lower and upper parts of $\mathcal{C}$ and relate them to the predicate $G$ for representing $\sigma$. 

For two events $(O_1,O_2)$ and $(V_1,V_2)$ and rational numbers
$p,q \in (0,1)$, with $p\leq q$, we have the {\em lower end} predicate
$C_\sigma^-((V_1,V_2),(O_1,O_2);p)$ and the {\em upper end} predicate
$C_\sigma^+((V_1,V_2),(O_1,O_2);q)$ with the intended meaning given,
respectively, by:
\[p< \sigma(V_1\cap O_1)/(1-\sigma (O_2)),\quad 1-\sigma(V_2\cap O_1)/(1-\sigma (O_2))<q.\] The rules relating $G$ with $C^-$ and $C^+$ are as follows. We assume $\operatorname{Pos}(O_1,O_2)$ for all the rules below.

\begin{enumerate}
   \item[(L1)] 
   $ G_\sigma(V_1\cap O_1;p_1) \ \wedge \ $ $ G_\sigma(O_2;p_2)\implies C_\sigma^-\left ((V_1,V_2),(O_1,O_2);\frac{p_1}{1-p_2}\right)$. 
   \item[(L2)] $ C_\sigma^-((V_1,V_2),(O_1,O_2);p_1) \implies  G_\sigma(O_2; p_2) \ \vee \ G_\sigma(V_1\cap O_1; p_1 \, (1 - p_2))$. 
   \item[(U1)] $ G_\sigma(V_2\cap O_1;p_1) \ \wedge \
     G_\sigma(O_2;p_2) \implies C_\sigma^+\left((V_1,V_2),(O_1,O_2);1 - \frac{p_1}{1-p_2}\right)$. 
   \item[(  U2)] $ C_\sigma^+((V_1,V_2),(O_1,O_2);q) \implies G_\sigma(O_2;p) \vee G_\sigma(V_2\cap O_1;(1-q)(1-p))$. 
\end{enumerate}
It is easy to check that the four rules L1, L2, U1, and U2 are sound given the intended meaning of $C_\sigma^-$ and $C_\sigma^+$.

\begin{theorem} ({\bf Completeness for conditional probability})
With the four rules  L1, L2, U1, and U2, the predicates $C^-$ and $C^+$ represent the conditional probability map $\mathcal{C}$ as follows:   
\[\frac{\sigma(V_1\cap O_1)}{1-\sigma(O_2)}=\sup\{p: C_\sigma^-(V,(O_1,O_2);p)\}\]
\[1-\frac{\sigma(V_2\cap O_1)}{1-\sigma(O_2)}=\inf\{q: C_\sigma^+(V,(O_1,O_2);q)\}\]
\end{theorem}

\begin{proof}
 For the lower predicate $C_\sigma^-$, we first check that:
\[\frac{\sigma(V_1\cap O_1)}{1-\sigma(O_2)}\leq\sup\{p: C_\sigma^-(V,(O_1,O_2);p)\}.\]
Let $p<\sigma(V_1\cap O_1)/(1-\sigma(O_2))$. Then by the continuity of division of a number by a non-zero number, there exists $p_1<\sigma(V_1\cap O_1)$ and $p_2< \sigma(O_2)$ such that $p< p_1/(1-p_2)$, and the result follows from (L1).

For the reverse direction, if \[\frac{\sigma(V_1\cap O_1)}{1-\sigma(O_2)}<\sup\{p: C_\sigma^-(V,(O_1,O_2);p)\},\]
then there exists $p$ such that $C_\sigma^-(V,(O_1,O_2),p)$ and \[\frac{\sigma(V_1\cap O_1)}{1-\sigma(O_2)}<p.\]
By the continuity of division of a number by a non-zero number, there exists a rational $q\in [0,1]$ close to $\sigma(O_2)$ with $\sigma(O_2)< q$ such that $\sigma(V_1\cap O_1)<p(1-q)$. 

By (L2), from $C_\sigma^-(V,(O_1,O_2),p)$ one can derive:
$ G_\sigma(O_2; q) \ \vee \ G_\sigma(V_1\cap O_1; p \, (1 - q))$,
which contradicts the predicate relations $ G_\sigma(O_2; q)$ and
$G_\sigma(V_1 \cap O_1, p(1-q))$, i.e., $q < \sigma(O_2)$,
$p(1-q)<\sigma(V_1\cap O_1)$.


Similarly, we show that \[\inf\{q:C_\sigma^+((V_1,V_2),(O_1,O_2);q)\}=1-\sigma(V_2\cap O_1)/(1-\sigma (O_2)).\]
In fact, this follows immediately from the result on the left end point since the right end point is simply one minus the left end point with $V_1$ replaced with $V_2$.
\end{proof}

We write $C_\sigma((V_1,V_2),(O_1,O_2);[p,q])$ if and only if:
\begin{equation*}
  C_\sigma^-((V_1,V_2),(O_1,O_2);p)\wedge C_\sigma^+((V_1,V_2),(O_1,O_2);q). 
\end{equation*}
For events $O$ and $V$, we now write:
\begin{equation*}
  \sigma(V|O):=\mathcal{C}_\sigma(V,O) = [p,q].
\end{equation*}

\begin{example}
\label{example:classical_vs_interval_prob}
Consider a sensor measuring a continuous value $x \in [0,1]$. The sensor outputs a ``high'' reading under the following conditions:
\begin{itemize}
    \item {Definitely high:} if $x \in (0.6, 1)$ (normal operation range).
    \item {Never high:} if $x \in (0.1, 0.6)$ (calibrated range where false positives are impossible).
    \item {Unknown:} if $x \in (0, 0.1)$ (possible false positives due to a sensor fault).
\end{itemize}

We want to assess the hypothesis that $x$ is in a \textbf{critical range} $[0.7, 0.9]$, but our knowledge of this range is also imprecise:
\begin{itemize}
    \item {Definitely critical:} if $x \in (0.8, 0.85)$.
    \item {Definitely not critical:} if $x \in (0, 0.7) \cup (0.95, 1)$.
    \item {Unknown:} for $x \in [0.7, 0.8] \cup [0.85, 0.95]$.
\end{itemize}

We assume a uniform prior probability measure $\sigma$ on $[0,1]$, so for any set $A$, $\sigma(A)$ is its Lebesgue measure (length).

Let the event $E = (E_1, E_2)$ be $E_1 = (0.6, 1)$ and $E_2 = (0.1, 0.6)$. $E_1$ is the set of points that are \emph{definitely} in the evidence, while $E_2$ is the set of points that are \emph{definitely not} in the evidence. Thus, the evidence is any set $B$ satisfying: 
\[
E_1 \subseteq B \subseteq [0,1] \setminus E_2 = [0, 0.1] \cup [0.6, 1].
\]

Let the hypothesis $H = (H_1, H_2)$ be $H_1 = (0.8, 0.85)$ and $H_2 = (0, 0.7) \cup (0.95, 1)$. The hypothesis is any set $A$ satisfying: 
\[
H_1 \subseteq A \subseteq [0,1] \setminus H_2 = [0.7, 0.95].
\]

The interval conditional probability of $H$ given $E$ is:
\[
C_\sigma(H, E) = \left[ \frac{\sigma(H_1 \cap E_1)}{1 - \sigma(E_2)},\; 1 - \frac{\sigma(H_2 \cap E_1)}{1 - \sigma(E_2)} \right].
\]
We have:
\begin{itemize}
    \item $\sigma(E_2) = \text{length}(0.1, 0.6) = 0.5, \quad 1 - \sigma(E_2) = 0.5$.

\item $H_1 \cap E_1 = (0.8, 0.85), \quad \sigma(H_1 \cap E_1) = 0.05$.

\item $H_2 \cap E_1 = (0.6, 0.7) \cup (0.95, 1), \quad
\sigma(H_2 \cap E_1) = 0.15$.
\end{itemize}
As such, we obtain the lower endpoint:
    \[
    \frac{\sigma(H_1 \cap E_1)}{1 - \sigma(E_2)} = \frac{0.05}{0.5} = 0.1,
    \]
    and the upper endpoint:
    \[
    1 - \frac{\sigma(H_2 \cap E_1)}{1 - \sigma(E_2)} = 1 - \frac{0.15}{0.5} = 1 - 0.3 = 0.7.
    \]
    which results in the interval conditional probability of:
\[
{C_\sigma(H, E) = [0.1, 0.7]}.
\]

In classical probability, one must choose precise sets for the evidence and the hypothesis. For instance, one might have $B = (0, 0.1) \cup (0.6, 1)$ as evidence and $A = (0.7, 0.9)$ as the hypothesis.

Then, the classical conditional probability is calculated as:
\begin{equation*}
  P(A \mid B) = \frac{P(A \cap B)}{P(B)} = \frac{\text{length}\big((0.7, 0.9) \cap ((0,0.1) \cup (0.6,1))\big)}{\text{length}((0,0.1) \cup (0.6,1))}  = \frac{0.2}{0.5} = 0.4.  
\end{equation*}
Thus, the classical approach yields a single point estimate:
\[
{P_{\text{classical}}(H \mid E) = 0.4} \in [ 0.1, 0.7].
\]

\begin{table*}[t]
\centering
\begin{tabular}{@{}lll@{}}
\toprule
\textbf{Aspect} & \textbf{Classical Probability} & \textbf{Interval Probability} \\ \midrule
Result & Point estimate: 0.4 & Interval: [0.1, 0.7] \\
Representation of uncertainty & Hides uncertainty by precision & Explicitly shows uncertainty \\
Interpretation of evidence & Forces $B = (0,0.1) \cup (0.6,1)$ & Allows any $B$ with $E_1 \subseteq B \subseteq [0,1]\setminus E_2$ \\
Interpretation of hypothesis & Forces $A = [0.7, 0.9]$ & Allows any $A$ with $H_1 \subseteq A \subseteq [0,1]\setminus H_2$ \\
Information content & Appears precise but is arbitrary & Honestly reflects imprecision \\
Decision-making & May lead to overconfidence & Supports robust decisions (e.g., worst-case analysis) \\ \bottomrule
\end{tabular}
\caption{Comparison between classical and interval probability approaches (Example~\ref{example:classical_vs_interval_prob}).}
\end{table*}
\end{example}

The conditional probability can be extended to credal sets on a Hausdorff metric space space $X$ as follows. \[\mathcal{C}:U(P(X))\times {\mathbb{E}(X)}\times {\mathbb{E}(X)} \to \mathbb{I}[0,1]\]
\[\widehat{\mathcal{C}}(K,((V_1,V_2),(O_1,O_2)) = \left[\inf_{\sigma\in K}\frac{\sigma(V_1\cap O_1)}{1-\sigma(O_2)},1-\inf_{\sigma\in K}\frac{\sigma(V_2\cap O_1)}{1-\sigma(O_2)}\right],\]
provided $\sigma(O_1)>0$ (which implies $1-\sigma(O_2)>0$) for all $\sigma\in K$. For a convex polytope, taking infimum is reduced to taking minimum over all vertices.

\begin{example}
\label{example:credal_vs_class_Bayesian}

Let \(X = [0,1]\) and define the events:
\[
O_1 = (0.3,0.7), \quad O_2 = (0,0.1)\cup(0.9,1), \quad O = (O_1, O_2),
\]
\[
V_1 = (0.4,0.6), \quad V_2 = (0,0.2)\cup(0.8,1), \quad V = (V_1, V_2).
\]

Consider the credal set \(K = \{\sigma_1, \sigma_2, \sigma_3, \sigma_4\}\) in which:
\[
\sigma_1 = \mathrm{Beta}(2,5),\ \sigma_2 = \mathrm{Beta}(5,2),\ \sigma_3 = \mathrm{Beta}(3,3),\ \sigma_4 = \mathrm{Beta}(0.5,0.5).
\]
The relevant numerical values taken by these distributions are presented in Table~\ref{table:sigma_values}.

\begin{table}[h]
\centering
\begin{tabular}{cccc}
\toprule
\(\sigma\) & \(\sigma(O_2)\) & \(\sigma(V_1 \cap O_1)\) & \(\frac{\sigma(V_1 \cap O_1)}{1-\sigma(O_2)}\) \\
\midrule
\(\sigma_1\) & 0.085 & 0.525 & 0.574 \\
\(\sigma_2\) & 0.410 & 0.498 & 0.844 \\
\(\sigma_3\) & 0.01712 & 0.365 & 0.371 \\
\(\sigma_4\) & 0.4096 & 0.239 & 0.405 \\
\bottomrule
\end{tabular}
\caption{Numerical values for Example~\ref{example:credal_vs_class_Bayesian}}
\label{table:sigma_values}
\end{table}

\emph{Credal Set Conditional Probability:} For each \(\sigma \in K\),
the values $\sigma(O_2)$ and $\sigma(V_1 \cap O_1)$ are given in
Table~\ref{table:sigma_values}. Since \(\sigma(V_2 \cap O_1) = 0\) for
all \(\sigma\), the conditional probability interval is:
\[
\hat{C}(K, (V, O)) = 
\left[ \inf_{\sigma \in K} \frac{\sigma(V_1 \cap O_1)}{1 - \sigma(O_2)},\ 1 \right]
= [0.371,\ 1].
\]

\emph{Classical Conditional Probability:} Take the mixture prior:
\[
\sigma_{\text{avg}} = \frac{1}{4} \sum_{i=1}^4 \sigma_i.
\]
We have $\sigma_{\text{avg}}(O_2) \approx 0.23043$, and $\sigma_{\text{avg}}(V_1 \cap O_1) \approx 0.40675$. The classical conditional probability is calculated as:
\[
P_{\text{classical}}(V \mid O) = 
\frac{0.40675}{1 - 0.23043} \approx 0.528.
\]

\begin{table*}[t]
\centering
\begin{tabular}{p{0.4\textwidth}p{0.25\textwidth}p{0.25\textwidth}}
\toprule
\textbf{Aspect} & \textbf{Credal Set (Interval)} & \textbf{Classical Bayesian} \\
\midrule
Prior form & Set of distributions \(K\) & Single distribution \(\sigma_{\text{avg}}\) \\
Output & Interval \([0.371, 1]\) & Point estimate \(0.528\) \\
Uncertainty preserved & Yes (width \(\approx 0.629\)) & No (collapsed) \\
Robustness to prior choice & High & Low \\
Decision safety & Conservative lower bound 0.371 & May over/under-estimate risk \\
\midrule
Classical point inside interval? & \multicolumn{2}{c}{Yes, \(0.528 \in [0.371, 1]\)} \\
\bottomrule
\end{tabular}
\caption{Comparison of the domain-theoretic conditional probability based on credal sets against the classical conditional probability (Example~\ref{example:credal_vs_class_Bayesian}).}
\label{table:credal_vs_class_Bayesian}
\end{table*}

The credal set method produces an interval that captures uncertainty across multiple priors, while the classical approach yields a single number that depends on an arbitrary choice of prior (here, the average). In safety‑critical applications, the interval's lower bound provides a robust, risk‑averse estimate, whereas the classical point estimate may misrepresent the true uncertainty. Further comparison is presented in Table~\ref{table:credal_vs_class_Bayesian}.
\end{example}

\section{Bayesian updating for events}

Recall the classical Bayesian rule for updating hypothesis $H$ with evidence $E$ and $\sigma$ as the probability distribution assuming $\sigma(E)\neq 0$:
\[\sigma(H|E)=\frac{\sigma(H)\sigma(E|H)}{\sigma(E)}=\frac{\sigma(H)\sigma(E|H)}{\sigma(H)\sigma(E|H)+\sigma(E|\neg H)(1-\sigma(H))}.\]

Let $x:=\sigma(H)$, $y:=\sigma(E|H)$, $z:=\sigma(E|\neg H)$; define $f_b:[0,1]^3\to [0,1]$ by 
\[f_b(x,y,z)=\frac{xy}{xy+z(1-x)}.\]
We then have for all $x,y,z\in [0,1]$:
\begin{equation*}
\frac{\partial f_b}{\partial x}=\frac{yz}{(xy+z(1-x))^2}\geq 0,\, \quad
\frac{\partial f_b}{\partial y}=\frac{x(1-x)z}{(xy+z(1-x))^2}\geq 0, \quad \frac{\partial f_b}{\partial z}=\frac{-x(1-x)y}{(xy+z(1-x))^2}\leq 0.
\end{equation*}

Thus, $f_b$ is increasing in $x$ and $y$, and is decreasing in $z$. It follows that $f_b$ has global minimum value $0$ attained at $x=y=0$ and $z=1$ and global maximum $1$ attained at $x=y=1$ and $z=0$. 
 Let $\hat{f_b}: (\mathbb{I}[0,1])^3\to \mathbb{I}[0,1]$ be the maximal extension of $\sigma$ to intervals:
   \[\hat{f_b}\left( \prod_{1\leq i\leq 3}[a_i,b_i]\right)= \left\{f_b(x,y,z): x\in [a_1,b_1],y\in [a_2,b_2],z\in[a_3,b_3]\right\}\]

\begin{lemma}\label{interval-bayes}
   \[\hat{f_b}\left( \prod_{1\leq i\leq 3}[a_i,b_i]\right)= \left[f_b(a_1,a_2,b_3),f_b(b_1,b_2,a_3)\right]\] 
\end{lemma}

 We will now extend the Bayesian rule to events. Consider two pairs of disjoint open sets $H=(H_1,H_2) $ and $E=(E_1,E_2)$ for the hypothesis and the evidence, respectively, and put $\neg H=(H_2,H_1)$.
We define $\mathcal{B}_\sigma:(\mathbb{E}(D))^3\to {\bf I}[0,1]$, using Lemma~\ref{interval-bayes}, with:

\begin{multline}
\label{interval-bayes-details}
\mathcal{B}_\sigma\left((H_1,H_2),((E|H)_1,(E|H)_2), (E|\neg H)_1,(E|\neg H)_2)\right) \\
 =\hat{f}_b(\sigma(H_1,H_2),\sigma(E|(H_1,H_2)),\sigma(E|(H_2,H_1)))\\=\bigl[f_b(p_h,p_{e|h},q_{e|\neg h}),\; f_b(q_h,q_{e|h},p_{e|\neg h})\bigr]\\=
\left[\frac{p_hp_{e|h}}{p_hp_{e|h}+q_{e|\neg h}(1-p_h)},\frac{q_hq_{e|h}}{q_hq_{e|h}+p_{e|\neg h}(1-q_h)}\right],    
\end{multline}
with the following
notations:
\begin{definition}
  \label{def:bayes-notation}
  Given hypothesis and evidence events
  $H=(H_1,H_2)$, $E=(E_1,E_2)\in\mathbb{E}(D)$ with $\sigma(E_1)>0$ (so
  that $\operatorname{Pos}(E)$ holds), define the following interval
  probabilities:
\begin{align}
    [p_h,q_h] &:= \sigma(H_1,H_2) = [\sigma(H_1), 1-\sigma(H_2)], \label{eq:ph}\\
    [p_{e|h},q_{e|h}] &:= \sigma(E|(H_1,H_2)) = C_\sigma(E,H), \label{eq:peh}\\
    [p_{e|\neg h},q_{e|\neg h}] &:= \sigma(E|(H_2,H_1)) = C_\sigma(E,\neg H), \label{eq:penoth}
\end{align}
where $C_\sigma$ is the conditional probability map from Equation~\ref{cond-prob-map-events} and $\neg H = (H_2,H_1)$.
\end{definition}
We can now derive the (sharp) interval Bayesian rule: \begin{corollary}\label{def-int-bayes-rule}
    The {\em interval Bayesian rule} is given by:
    \[\sigma((H_1,H_2)|E):=\bigl[f_b(p_h,p_{e|h},q_{e|\neg h}),\; f_b(q_h,q_{e|h},p_{e|\neg h})\bigr].\]
    \end{corollary}

\subsection{Bayesian Inference Rules}

We now formulate two predicates for the interval Bayesian method as we did for the interval conditional probability in Section~\ref{subsec:cond-prob-pred}. The two new predicates are $B^-(H,E;p)$ and $B^+(H,E;q)$, respectively, with the intended meaning, using the notations in Equation~\eqref{interval-bayes-details}  given, respectively, by:
\begin{equation}\label{bayes-law}
    p<\frac{p_hp_{e|h}}{p_hp_{e|h}+q_{e|\neg h}(1-p_h)},
    \end{equation}
    and
    \begin{equation}\label{bayes-law1}\frac{q_hq_{e|h}}{q_hq_{e|h}+p_{e|\neg h}(1-q_h)}<q. \end{equation}

Next, we present the four inference rules for interval Bayesian updating. Consider again the two pairs of disjoint open sets $H=(H_1,H_2) $ and $E=(E_1,E_2)$ for the hypothesis and the evidence, respectively, and let $\neg H:=(H_2,H_1)$. In the four rules below, for ease of reading,  we use the notation $p_h'$, $p_{e|h}'$ and $p'_{e|\neg h}$ for  rational numbers assumed to satisfy $p_h'<p_h$, $p_{e|h}'<p_{e|h}$, and $p'_{e|\neg h}<p_{e|\neg h}$. Similarly, we use $q_h'$, $q_{e|h}'$, and $q'_{e|\neg h}$ for rational numbers assumed to satisfy $q_h<q'_h$, $q_{e|h}<q'_{e|h}$ and $q_{e|\neg h}<q'_{e|\neg h}$.

\noindent
(B1)

\[ C^-(H,(D,\emptyset);p_1) \wedge C^-(E,H;p_2) \wedge C^+(E,\neg
  H;q_3) \implies B^-\left(H,E;\frac{p_1p_2}{p_1p_2+q_3(1-p_1)}\right).\]

\noindent
(B2)
\begin{equation*}
    B^-\left(H,E;p\right) \implies \exists p'_h,p'_{e|h},p'_{e|\neg
      h}\in\mathbb{Q}_0: \quad 
    \left(p<\frac{p'_hp'_{e|h}}{p'_hp'_{e|h}+q'_{e|\neg h}(1-p'_h)}\right) \wedge T_1 \wedge T_2 \wedge T_3.
\end{equation*}
where $T_1 = C^-(H,(D,\emptyset);p'_h)$, $T_2 = C^-(E,H;p'_{e|h})$, and $T_3 = C^+(E,\neg H;q'_{e|\neg h})$.

\noindent
(B3)\[C^+(H,(D,\emptyset);q_1) \wedge C^+(E,H;q_2) \wedge C^-(E,\neg H;p_3)\implies B^+\left(H,E;\frac{q_1q_2}{q_1q_2+p_3(1-q_1)}\right).\]

\noindent
(B4)
\begin{equation*}    B^+\left(H,E;q\right)\implies\exists
  q'_h,q'_{e|h},q'_{e|\neg h}\in\mathbb{Q}_0:  \quad
\left(\frac{q'_h q'_{e|h}}{q'_hq'_{e|h}+p'_{e|\neg h}(1-q'_h)}<q\right) \wedge T_1 \wedge T_2 \wedge T_3,
\end{equation*}
in which $T_1 = C^+(H,(D,\emptyset);q'_h)$, $T_2 = C^+(E,H;q'_{e|h})$,
and $T_3= C^-(E,\neg H;p'_{e|\neg h})$.

The soundness of the four rules (B1), (B2), (B3) and (B4) follows from Lemma~\ref{interval-bayes}. As for the completeness:

\begin{theorem}\label{bayes-complete}The Bayesian Rules (B1)--(B4) are complete.
    \end{theorem}
    \begin{proof}
          We need to check the predicate $B^-(H,E;p)$  satisfies:
           \begin{equation}\label{bayes-first}
                \sup\{p:B^-(H,E;p)\}=\frac{p_hp_{e|h}}{p_hp_{e|h}+q_{e|\neg h}(1-p_h)}.\end{equation}
           From (B1), we obtain: \[\sup\{p:B^-(H,E;p)\}\geq \frac{p_hp_{e|h}}{p_hp_{e|h}+q_{e|\neg h}(1-p_h)},\]
        and (B2) implies 
        \[\sup\{p:B^-(H,E;p)\}\leq \frac{p_hp_{e|h}}{p_hp_{e|h}+q_{e|\neg h}(1-p_h)},\]
        from which Equation~\eqref{bayes-first} follows. Similarly, we can show that:
        \[\inf\{q:B^+(H,E;q)\}=\frac{q_hq_{e|h}}{q_hq_{e|h}+p_{e|\neg h}(1-q_h)}.\]
    \end{proof}

\begin{example}
Consider a medical test for a disease. Let:
\begin{itemize}
    \item $H$: Hypothesis that the patient has the disease.
    \item $E$: Evidence that the test is positive.
\end{itemize}

\noindent
We have imprecise information:
\begin{enumerate}
    \item {Prior prevalence}: From epidemiological studies, the disease prevalence is estimated to be between 1\% and 5\%, but the exact value is uncertain.
    \item {Test sensitivity}: The probability that the test is positive given the disease is between 85\% and 95\%.
    \item {Test specificity}: The probability that the test is negative given no disease is between 90\% and 99\%.
\end{enumerate}

\noindent
In classical Bayesian analysis, point estimates are typically chosen:
\begin{align*}
    P(H) &= 0.03 \quad \text{(midpoint of 1\%-5\%)} \\
    P(E|H) &= 0.90 \quad \text{(midpoint of 85\%-95\%)} \\
    P(E|\neg H) &= 1 - 0.945 = 0.055 \quad \text{(where 94.5\% is midpoint of 90\%-99\%)}
\end{align*}

For a positive test result, the classical posterior is:
\begin{align*}
  P(H|E) &= \frac{P(H)P(E|H)}{P(H)P(E|H) + P(\neg H)P(E|\neg H)}\\
  &= \frac{0.03 \times 0.90}{0.03 \times 0.90 + 0.97 \times 0.055} = \frac{0.027}{0.027 + 0.05335} \approx 0.336.
\end{align*}
\noindent
Thus, the classical approach yields ${P_{\text{classical}}(H|E) \approx 33.6\%}$.

We now represent the imprecise information in the domain-theoretic framework as intervals:
\begin{align*}
    \sigma(H) &= [p_h, q_h] = [0.01, 0.05] \\
    \sigma(E|H) &= [p_{e|h}, q_{e|h}] = [0.85, 0.95] \\
    \sigma(E|\neg H) &= [p_{e|\neg h}, q_{e|\neg h}] = [0.01, 0.10] \end{align*}
    since specificity is [0.90, 0.99].

The interval posterior is given by:
\[
\sigma(H|E) = \left[ \frac{p_h p_{e|h}}{p_h p_{e|h} + q_{e|\neg h}(1-p_h)},\ 
                     \frac{q_h q_{e|h}}{q_h q_{e|h} + p_{e|\neg h}(1-q_h)} \right]
\]

\begin{enumerate}
    \item \textbf{Lower bound (most conservative)}:
    \begin{equation*}
        \text{Lower} = \frac{0.01 \times 0.85}{0.01 \times 0.85 + 0.10 \times (1-0.01)} 
        = \frac{0.0085}{0.1075} \approx 0.0791
    \end{equation*}
    
    \item \textbf{Upper bound (most optimistic)}:
    \begin{align*}
        \text{Upper} &= \frac{0.05 \times 0.95}{0.05 \times 0.95 + 0.01 \times (1-0.05)} = \frac{0.0475}{0.057} \approx 0.8333
    \end{align*}
\end{enumerate}

\noindent
Thus, the interval posterior is:
\[
{\sigma(H|E) = [0.0791, 0.8333] \approx [7.9\%, 83.3\%]}.
\]
Further comparisons between the interval and the classical Bayesian updating are presented in Table~\ref{table:classical_interval_Bays}.

\begin{table*}[t]
\centering
\begin{tabular}{@{}lll@{}}
\toprule
\textbf{Aspect} & \textbf{Classical Bayesian} & \textbf{Interval Bayesian} \\
\midrule
\textbf{Result} & Point: 33.6\% & Interval: [7.9\%, 83.3\%] \\
\textbf{Uncertainty representation} & Hidden in point estimate & Explicit in interval width \\
\textbf{Information used} & Midpoints only & Entire ranges \\
\textbf{Decision implications} & Clear-cut (33.6\%) & Ambiguous (could be 8\% or 83\%) \\
\textbf{Robustness} & Sensitive to midpoint choice & Robust to parameter variations within intervals \\
\bottomrule
\end{tabular}
\caption{Comparison of classical vs. interval Bayesian updating.}
\label{table:classical_interval_Bays}
\end{table*}

    \end{example}

\subsection{Bayesian Updating for Credal sets}
    
  Finally, for a compact metric space $X$, as with $\hat{\mathcal{E}}$ and $\hat{\mathcal{C}}$, we obtain the extension of $B$ on all its inputs:
  \[\hat{\mathcal{B}}:U(P(X))\to (\mathbb{E}(X))^3\to \mathbb{I}[0,1]\]
  with
  \begin{equation*}
    \hat{\mathcal{B}}(K,(H,E|H,E|\neg H))=\left[\inf_{\sigma\in K}\frac{p^\sigma_hp^\sigma_{e|h}}{p^\sigma_hp^\sigma_{e|h}+q^\sigma_{e|\neg h}(1-p^\sigma_h)},\sup_{\sigma\in K}\frac{q^\sigma_hq^\sigma_{e|h}}{q^\sigma_hq^\sigma_{e|h}+p^\sigma_{e|\neg h}(1-q^\sigma_h)}\right],
  \end{equation*}
  with the following
  notations: \[\sigma(H_1,H_2)=[p^\sigma_h,q^\sigma_h],\quad\sigma(E|(H_1,H_2))=[p^\sigma_{e|h},q^\sigma_{e|h}],
    \quad \sigma(E|(H_2,H_1))=[p^\sigma_{e|\neg h},q^\sigma_{e|\neg h}].\]
  For a convex polytope $K\in U(P(X))$ with a finite set $N\subseteq K$ of nodes, we have:

  \begin{equation}\label{credal-Bayesian}
    \hat{\mathcal{B}}(K,(H,E|H,E|\neg H)) =\left[\min_{\sigma\in N}\frac{p^\sigma_hp^\sigma_{e|h}}{p^\sigma_hp^\sigma_{e|h}+q^\sigma_{e|\neg h}(1-p^\sigma_h)},\max_{\sigma\in N}\frac{q^\sigma_hq^\sigma_{e|h}}{q^\sigma_hq^\sigma_{e|h}+p^\sigma_{e|\neg h}(1-q^\sigma_h)}\right].
  \end{equation}

\begin{example}
Consider a family of distributions $\sigma_{a,b}$ indexed by two parameters $a,b \in [0,1]$, where:
\begin{align*}
    \sigma_{a,b}(H) &= [0.1 + 0.2a,\; 0.3 + 0.2a], \\
    \sigma_{a,b}(E|H) &= [0.7 + 0.2b,\; 0.9 + 0.1b], \\
    \sigma_{a,b}(E|\neg H) &= [0.05 + 0.1b,\; 0.15 + 0.1b].
\end{align*}
Thus,
\begin{align*}
    p_h^{a,b} &= 0.1 + 0.2a, & q_h^{a,b} &= 0.3 + 0.2a, \\
    p_{e|h}^{a,b} &= 0.7 + 0.2b, & q_{e|h}^{a,b} &= 0.9 + 0.1b, \\
    p_{e|-h}^{a,b} &= 0.05 + 0.1b, & q_{e|-h}^{a,b} &= 0.15 + 0.1b.
\end{align*}
\noindent
Applying the interval Bayesian rule, we compute:

\paragraph{Lower bound:}
Minimisation over $a,b$ yields $a=0, b=1$:
\[
\frac{(0.1)(0.9)}{(0.1)(0.9) + (0.25)(0.9)} = \frac{0.09}{0.315} \approx 0.2857.
\]

\paragraph{Upper bound:}
Maximisation over $a,b$ yields $a=1, b=0$:
\[
\frac{(0.5)(0.9)}{(0.5)(0.9) + (0.05)(0.5)} = \frac{0.45}{0.475} \approx 0.9474.
\]
Thus, $\hat{B}(K, H, E) \approx [0.2857,\; 0.9474]$.

\paragraph{Classical Bayesian approach:}
In the classical setting, one must choose precise parameters. Taking midpoints $a=0.5, b=0.5$:
\begin{equation*}
        \sigma(H) \approx 0.3, \quad
    \sigma(E|H) \approx 0.875, \quad
    \sigma(E|\neg H) \approx 0.15.
\end{equation*}
Then:
\[
P(H|E) = \frac{0.3 \times 0.875}{0.3 \times 0.875 + 0.15 \times 0.7} \approx 0.7143\in [0.2857,\; 0.9474].
\]


The classical point estimate lies inside the interval, but fails to capture the full uncertainty. The interval width of $0.6617$ reflects substantial sensitivity to the parameters $a$ and $b$. In safety‑critical applications, this reveals that the posterior probability could be as low as $28.6\%$ or as high as $94.7\%$, information that is lost when a single precise parameter set is chosen arbitrarily.
\end{example}

\section{Extension to multi-dimensional case}
All the results in previous sections can be generalised to the multi-dimensional case. Let $D_1$ and $D_2$ be basic spaces. Then, $D_1\times D_2$ is again a basic space with a countable basis of its lattice of open sets given by $b_1\times b_2$ with $b_i\in B_i$ with $i=1,2$. 

The predicate $G$ extends to the product space $D_1\times D_2$ as $G(U\times V;p)$ with the intended meaning $p< \sigma(U\times V)$. With the same axioms interpreted more generally, we obtain the joint probability distribution on the product space $D_1\times D_2$. 
For $U\in \mathbb{E}(D_1)$ and $V\in \mathbb{E}(D_2)$ with joint distribution $\sigma$, the conditional probability of $U$ given $V$ now takes the form:
\[\mathcal{C}_\sigma(U,V)=\left[\frac{\sigma(U_1\times V_1)}{1-\sigma(V_2)},1-\frac{\sigma(U_2\times V_1)}{1-\sigma(V_2)} \right]\]
where $\sigma(V_2)=\sigma(D_1\times V_2)$ is the marginal distribution of $V_2$.
Again, the predicates $C^-$ and $C^+$ are in the more general form $C^-(U\otimes V,U'\otimes V'; p)$ and  $C^+(U\otimes V,U'\otimes V'; q)$ with the rules L1, L2, U1, and U2 in terms of the multi-dimensional predicates $G$, $C^-$, and $C^+$. Soundness and completeness theorems 
can be obtained as before.

\section{Conditional Independence}

In this section, let $U\in \mathbb{E}(D_1)$, $V\in \mathbb{E}(D_2)$, $W\in \mathbb{E}(D_3)$, and $Z\in \mathbb{E}(D_4)$ be events in the  basic spaces $D_i$, for $1\leq i\leq 4$.

For the three events $U=(U_1,U_2)$, $V=(V_1,V_2)$, and $W=(W_1,W_2)$, we define {\em the conditional independence of the event $U$ of $V$ given $W$}, with respect to their joint probability distribution. In order to be able to use the factorisation properties of conditional independence in classical probability theory, we will assume that the event $(W_1,W_2)\in \mathbb{E}(D_3)$ on which conditioning takes place is a classical event, i.e., $\sigma(W_1)+\sigma(W_2)=1$. We introduce the predicate $\mbox{CE}(W_1,W_2)$ with the intended meaning $\sigma(W_1)+\sigma(W_2)=1$. With the assumption $\mbox{CE}(W_1,W_2)$, the conditional probability simplifies to:
\begin{equation}
  \label{CE-sim}
  C(\sigma,(U_1,U_2),(W_1,W_2))= \left[ \frac{\sigma(U_1\cap W_1)}{1-\sigma(W_2)}, 1-\frac{ \sigma(U_2\cap W_1)}{1-\sigma(W_2)}\right]=\left[ \frac{\sigma(U_1\cap W_1)}{\sigma(W_1)}, 1-\frac{ \sigma(U_2\cap W_1)}{\sigma(W_1)}\right].
\end{equation}

This enables us to define: $U\dperp V \mid W$ if $U_1\bot V_1 \mid W_1$, i.e., if $U_1$ is independent of $V_1$ given $W_1$ as in classical probability theory. We now define the predicate $I(U\dperp V \mid W)$ with the  intended meaning of $U\dperp V \mid W$, i.e., $U_1\bot V_1 \mid W_1$, equivalently $\sigma(U_1\times V_1 \mid W_1)=\sigma(U_1 \mid W_1)\sigma(V_1 \mid W_1)$. This definition leads to the four rules of conditional independence extended to events. 

\begin{enumerate}
    \item[(CI1)] {\bf Symmetry:} $\operatorname{CE}(W) \wedge I(U\dperp V \mid W)\implies I(V\dperp U \mid W)$.

    \item[(CI2)] {\bf Weak union:} 
      $\operatorname{CE}(Z) \wedge \operatorname{CE}(W) \wedge
      I(U\dperp (V\otimes W) \mid Z) \implies I(U\dperp V \mid
      W\otimes Z)$.
     \item[(CI3)] {\bf Contraction:} $\operatorname{CE}(V) \wedge \operatorname{CE}(Z) \wedge I(U\dperp V  \mid Z) \wedge I(U\dperp W|V\otimes Z)
     \implies I(U\dperp V\otimes W \mid Z).$
\end{enumerate}
If the  joint probability in $D_1\otimes D_2\otimes D_3\otimes D_4$ is strictly positive for any open set $O\in \mathbb{B}_{D_1\otimes D_2\otimes D_3\otimes D_4}$, we will have the intersection property which is stronger than contraction. 
Thus, we can state:
\begin{enumerate}
    \item[(CI4)] {\bf Intersection:} For all $U\otimes V\otimes W\otimes Z\in B_{\mathbb{A}}$, where $\mathbb{A}=\mathbb{E}(D_1)\otimes \mathbb{E}(D_2)\otimes \mathbb{E}(D_3)\otimes \mathbb{E}(D_4)$:
    \begin{multline*}
      \Big(\operatorname{Pos}(U_1\times V_1\times W_1\times Z_1)
      \wedge \operatorname{CE}(V) \wedge \operatorname{CE}(W) \wedge
      \operatorname{CE}(Z) \wedge I(U\dperp V \mid W\otimes Z) \wedge
      I(U\dperp W \mid V\otimes Z) \Big) \\ \implies I(U\dperp V\otimes W \mid Z).  
    \end{multline*}
\end{enumerate}
We also need rules for the predicate $I(U\dperp V \mid W)$ as follows:
\begin{enumerate}
    \item[(CI5)] \[\operatorname{Pos}(W) \wedge \operatorname{CE}(W) \wedge I(U\dperp V \mid W) \wedge C^-(U,W;p_U) \wedge C^-(V,W;p_V)\]\[\implies C^-(U\otimes V,W;p_Up_V).\] 
    \item[(CI6)]
    \begin{multline*}
    \operatorname{Pos}(W)\wedge \operatorname{CE}(W)\wedge I(U\dperp V \mid W)\wedge  C^-(U\otimes V,W;p)\\    
    \implies \exists p_1,p_2\in \mathbb{Q}_0: \big( p<p_1p_2 \big) \wedge C^-(U,W;p_1)\wedge C^-( V,W;p_2).
    \end{multline*}  
\end{enumerate}

The rules (CI5) and (CI6) ensure that our predicates correctly compute the left-end point of $C^-(U\otimes V,W;p_Up_V)$ based on Equation~(\ref{CE-sim}) as the product of the left-end points of $\mathcal{C}_\sigma(U,W)$ and $\mathcal{C}_\sigma(V,W)$. 

It remains to handle the right-end point soundly. We have:
\begin{equation}\label{rep}\mathcal{C}_\sigma^+(U\otimes V \mid W)=
  1-\sigma((U_2\times D_V) \cup (D_U\times V_2) \mid  W_1)
  =1-(\sigma(U_2 \mid W_1)+\sigma(V_2 \mid W_1))+\sigma(U_2\times V_2
  \mid W_1))
\end{equation}

We can use Fr{\'e}chet's rule to conservatively estimate the right end point.

\begin{lemma}
  \label{lemma:Frechet}
  If $\operatorname{CE}(W)$, then the Fr{\'e}chet rule
  implies:\[\mathcal{C}_\sigma^+(U\otimes V \mid W)\leq
    \min\{1-\sigma(U_2 \mid W_1),1-\sigma(V_2 \mid W_1)\}.\]
\end{lemma}

\begin{proof}
As we have $\operatorname{CE}(W)$, we obtain
\[\sigma(U \mid W)=\mathcal{C}_\sigma(U,W)=[\sigma(U_1 \mid W_1), 1-\sigma(U_2 \mid W_1)],\]
\[\sigma(V \mid W)=\mathcal{C}_\sigma(V,W)=[\sigma(V_1 \mid W_1), 1-\sigma(V_2 \mid W_1)].\]
By assuming $\operatorname{CE}(W)$, since $U\otimes V=(U_1\times V_1, (U_2\times  D_V)\cup (D_U\times V_2))$, we have:
\begin{equation}
\label{conditional-product}\mathcal{C}_\sigma(U\otimes V,W)=  \left[\sigma(U_1\times V_1 \mid  W_1),1-\sigma((U_2\times D_V) \cup (D_U\times V_2) \mid  W_1)\right].
\end{equation}
The last equality yields, using $q_U:=1-\sigma(U_2 \mid W_1)$ and $q_V:=1-\sigma(V_2 \mid W_1)$:
\[\begin{array}{ll}\mathcal{C}_\sigma^+(U\otimes V \mid W)\\
= 1-(\sigma(U_2 \mid W_1)+\sigma(V_2 \mid W_1)+\sigma(U_2\times V_2 \mid W_1))\\
=1-(1-q_U+1-q_V)+\sigma(U_2\times V_2 \mid W_1))\\=q_U+q_V-1+\sigma(U_2\times V_2 \mid W_1)\\
\leq q_U+q_V-1+\min(\sigma(U_2 \mid W_1),\sigma(V_2 \mid W_1))&(\text{Fr{\'e}chet's rule})\\
=q_U+q_V-1+\min(1-q_U,1-q_V)\\
=q_U+q_V-\max(q_U,q_V)\\
=\min{(q_U,q_V)}\end{array}\]
where we have used Fr{\'e}chet's rule \[\sigma(U_2\times V_2 \mid W_1)\leq \min(\sigma(U_2 \mid W_1),\sigma(V_2 \mid W_1)).\]
\end{proof}

We thus include two new rules to compute the right-end point of $\mathcal{C}_\sigma(U\otimes V \mid W)$ as follows:

\begin{enumerate}
     \item[(CI7)] \[\operatorname{Pos}(W) \wedge  \operatorname{CE}(W) \wedge I(U\dperp V \mid W) \wedge C^+(U,W;q_U) \wedge  C^+(V,W;q_V)\]\[\implies C^+(U\otimes V,W;q),\] 
    where $q=\min(q_U,q_V)$.
    \item[(CI8)] 
    \begin{multline*}
      \operatorname{Pos}(W)\wedge \operatorname{CE}(W)\wedge I(U\dperp V \mid W)\wedge  C^+(U\otimes V,W;q)  \\
      \implies 
      \exists q_1,q_2\in \mathbb{Q}_0: \big( \min(q_1,q_2)<q \big) \wedge C^+(U,W;q_1)\wedge C^+( V,W;q_2).
    \end{multline*}    
\end{enumerate}

We can then prove:

\begin{theorem}
\label{s-c-c-i}
We have soundness and completeness for conditional independence
\end{theorem}

\begin{proof}
Soundness follows from the fact that all the rules are sound rules in classical probability theory. To verify completeness,  the conditional probability of $U$ given $W$ in Section 4, when $1-\sigma(W_2)=\sigma(W_1)>0$, satisfies:
\begin{equation}
\label{CP-1}
   a^-:= \frac{\sigma(U_1\times W_1)}{\sigma(W_1)}=\sup\{p: C^-(U,W;p)\}.
\end{equation}
Similarly, for the conditional probability of $V$ given $W$ and that of $U\otimes V$ given $W$, we have, respectively:
\begin{equation}\label{CP-2}b^-:=\frac{\sigma(V_1\times W_1)}{\sigma(W_1)}=\sup\{p: C^-(V,W;p)\},
\end{equation}
\begin{equation}
\label{CP-3}
c^-:=\frac{\sigma(U_1\times V_1\times W_1)}{\sigma(W_1)}=\sup\{p: C^-(U\otimes V,W;p)\}.
\end{equation}
Thus, we need to show that the product of the RHS of Equations~(\ref{CP-1}) and~(\ref{CP-2}) is the same as the RHS of Equation~(\ref{CP-3}), i.e., $a^-b^-=c^-$. From (CI6) it follows that $a^-b^-\leq c^-$, while from (CI7) 
it follows that $c^-\leq a^-b^-$.
Next we consider the right end point. 
\begin{equation}
\label{CP-1+}
   a^+:= 1-\frac{\sigma(U_2\times W_1)}{\sigma(W_1)}=1-\sigma(U_2 \mid W_1)=\inf\{q: C^+(U,W;q)\},
   \end{equation}
\begin{equation}
\label{CP-2+}b^+:=1-\frac{\sigma(V_2\times W_1)}{\sigma(W_1)}=1-\sigma(V_2 \mid W_1)=\inf\{q: C^+(V,W;q)\},
\end{equation}
\begin{equation}
c^+:=\inf\{q: C^+(U\otimes V,W;q)\}.
\end{equation}
With the Fr{\'e}chet conservative estimate, we need to check that $\min(a^+,b^+)=c^+$. From (CI7), it follows that $c^+\leq \min(a^+,b^+)$, whereas (CI8) implies that $ \min(a^+,b^+)\leq c^+$.
\end{proof}

\subsection{Strong conditional independence}

 Earlier in this section, we saw that classical conditional independence implies the lower conditional support factorises when two independent events $U$ and $V$ are conditioned to $W$. In this domain-theoretic setting we have additional information given by the upper conditional support. 

 \begin{lemma} Assuming $\operatorname{Pos}(W_1)$, we have:
    \[\mathcal{C}_\sigma(U\otimes V,W)=\mathcal{C}_\sigma(U,W) \cdot\mathcal{C}_\sigma(V,W) \iff\]
    
    \begin{itemize}
        \item $\sigma(U_1\times V_1 \mid W_1)=\sigma(U_1 \mid W_1)\sigma(V_1 \mid W_1)$.
        \item $\sigma(U_2\times V_2 \mid W_1)=\sigma(U_2 \mid W_1)\sigma(V_2 \mid W_1)$.
    \end{itemize}
 \end{lemma}
 \begin{proof}
    We already know that the equality of the left-end point of the interval equation is given by the first condition. For the equality of the right-end point, from Equation~\eqref{rep}, we obtain:
    
    \[1-(\sigma(U_2 \mid W_1)+\sigma(V_2 \mid W_1))+\sigma(U_2\times V_2 \mid W_1))\]\[=(1-\sigma(U_2 \mid W_1))(1-\sigma(V_2 \mid W_1)\]
    which is equivalent to $\sigma(U_2\times V_2 \mid W_1)=\sigma(U_2 \mid W_1)\sigma(V_2 \mid W_1)$.    
 \end{proof}

 \begin{definition}
 \label{def:strong_cond_ind}
    We say the event $U$ is {\em strongly conditionally independent}, of $V$ given $W$, denoted as $U\dperp_s V \mid W$, if the following two conditions hold: \begin{itemize}
        \item $\sigma(U_1\times V_1 \mid W_1)=\sigma(U_1 \mid W_1)\sigma(V_1  \mid W_1)$.
        \item $\sigma(U_2\times V_2 \mid W_1)=\sigma(U_2 \mid W_1)\sigma(V_2 \mid W_1)$.
    \end{itemize}
 \end{definition}
 The additional assumption, on the right-end point, in strong conditional independence is restrictive and is not likely to hold in many applications. However, it leads to computational efficiency, as both end-points of a conditional probability would be obtained by taking the products of corresponding end-points. One can think of strong conditional independence as providing an optimistic rule in graphical models to compute the right-end points of conditional probabilities.

 For strong conditional independence, we have two additional rules which replace (CI7) and (CI8):
 \begin{enumerate}
     \item[(SI9)] \[\operatorname{Pos}(W) \wedge  \operatorname{CE}(W) \wedge  I(U\dperp_s V \mid W) \wedge C^+(U,W;q_U) \wedge C^+(V,W;q_V)\]\[\implies C^+(U\otimes V,W;q_Uq_V).\]

  \item[(SI10)] 
  \begin{multline*}
    \operatorname{Pos}(W) \wedge  \operatorname{CE}(W) \wedge I(U\dperp_s V \mid W) \wedge C^+(U\otimes V,W;q)  \\
    \implies 
    \exists q_1,q_2\in \mathbb{Q}_0: q_1q_2<q\wedge C^+(U,W;q_1)\wedge C^+( V,W;q_2).
  \end{multline*}
 \end{enumerate}
 We can then deduce as in Theorem~\ref{s-c-c-i} that $a^+b^+=c_s^+$, where we have defined:

\begin{equation*}
   a^+:=1- \frac{\sigma(U_2\times W_1)}{\sigma(W_1)}=\inf\{q: C^+(U,W;q)\},
\end{equation*}
\begin{equation*}    
b^+:=1-\frac{\sigma(V_2\times W_1)}{\sigma(W_1)}=\inf\{q: C^+(V,W;q)\},  
\end{equation*}

\begin{align*}
c_s^+ &:=1-\sigma(U_2 \mid W_1)-\sigma(V_2 \mid W_1)+\sigma(U_2\times V_2 \mid W_1) \\
&=\inf\{q: C^+(U\otimes V,W;q)\}.
\end{align*}

The proof of the following theorem is straightforward:
\begin{theorem}
Soundness and completeness hold for strong conditional independence.
\end{theorem}

\begin{example}
    Consider three events in separate spaces:
\begin{itemize}
    \item $U$: Patient has symptom A, with $U_1 = \text{"has A"}$, $U_2 = \text{"no A"}$
    \item $V$: Patient has symptom B, with $V_1 = \text{"has B"}$, $V_2 = \text{"no B"}$
    \item $W$: Patient has disease D, with $W_1 = \text{"has D"}$, $W_2 = \text{"no D"}$
\end{itemize}

We know:
\begin{align*}
    \sigma(W_1) &= [0.3, 0.4] \quad \text{(30-40\% prevalence)} \\
    \sigma(U_1 \mid W_1) &= [0.6, 0.8] \quad \text{(60-80\% sensitivity for symptom A)} \\
    \sigma(V_1 \mid W_1) &= [0.7, 0.9] \quad \text{(70-90\% sensitivity for symptom B)} \\
    \sigma(U_2 \mid W_1) &= [0.2, 0.4] \quad \text{(20-40\% specificity for symptom A)} \\
    \sigma(V_2 \mid W_1) &= [0.1, 0.3] \quad \text{(10-30\% specificity for symptom B)}
\end{align*}

We assume $U$ and $V$ are conditionally independent given $W$, but we have two different interpretations of what this means in the interval framework.

\subsubsection*{Classical Approach}
Take midpoint values:
\begin{align*}
    P(W_1) &= 0.35 \\
    P(U_1 \mid W_1) &= 0.7, \quad P(V_1 \mid W_1) = 0.8 \\
    P(U_2 \mid W_1) &= 0.3, \quad P(V_2 \mid W_1) = 0.2
\end{align*}

Assuming conditional independence:
\begin{align*}
    P(U_1 \cap V_1  \mid  W_1) &= 0.7 \times 0.8 = 0.56 \\
    P(U_2 \cap V_2  \mid  W_1) &= 0.3 \times 0.2 = 0.06 \\
    P((U_2 \times D_V) \cup (D_U \times V_2)  \mid  W_1) &= 0.3 + 0.2 - 0.06 = 0.44
\end{align*}

Thus:
\[
C_{\text{classical}}(U \otimes V, W) = [0.56, 1 - 0.44] = [0.56, 0.56]
\]

\subsubsection*{Interval Approach with Fr\'echet Rule}
Using conservative propagation (Rules I10, I11):
\begin{align*}
    C^{-}(U \otimes V, W) &= [0.6 \times 0.7, \ 0.8 \times 0.9] = [0.42, 0.72] \\
    C^{+}(U \otimes V, W) &= [\min(0.4, 0.3), \ \min(0.8, 0.9)] = [0.3, 0.8]
\end{align*}

By Fr\'echet's inequality:
\begin{equation*}
C^{+}_{\text{Fr\'echet}}(U \otimes V, W) =
1 - [\sigma(U_2 \mid W_1) + \sigma(V_2 \mid W_1) - \min(\sigma(U_2 \mid W_1), \sigma(V_2 \mid W_1))].  
\end{equation*}
Taking bounds:
\begin{align*}
\text{Lower bound for } C^{+} &= 1 - (0.4 + 0.3 - \min(0.4, 0.3))\\
&= 1 - (0.7 - 0.3) = 0.6 \\
\text{Upper bound for } C^{+} &= 1 - (0.2 + 0.1 - \min(0.2, 0.1)) \\
&= 1 - (0.3 - 0.1) = 0.8
\end{align*}
Thus:
\[
C_{\text{Fr\'echet}}(U \otimes V, W) = [0.42, 0.80]
\]

\subsubsection*{Interval Approach with Strong Conditional Independence}
Assuming both positive and negative evidence factorize (Definition 7.4):
\begin{align*}
    C^{-}_{\text{strong}}(U \otimes V, W) &= [0.6 \times 0.7, \ 0.8 \times 0.9] = [0.42, 0.72] \\
    C^{+}_{\text{strong}}(U \otimes V, W) &= [0.4 \times 0.3, \ 0.8 \times 0.9] = [0.12, 0.72]
\end{align*}

Thus:
\[
C_{\text{strong}}(U \otimes V, W) = [0.42, 0.72].
\]

\begin{table*}[t]
\centering
\begin{tabular}{@{}lccc@{}}
\toprule
\textbf{Approach} & \textbf{Interval} & \textbf{Width} & \textbf{Contains Classical?} \\
\midrule
Classical & [0.56, 0.56] & 0.00 & Yes (by definition) \\
Fr\'echet Rule & [0.42, 0.80] & 0.38 & Yes (0.56 in [0.42, 0.80]) \\
Strong Independence & [0.42, 0.72] & 0.30 & Yes (0.56 in [0.42, 0.72]) \\
\bottomrule
\end{tabular}
\caption{Comparison of conditional independence treatments}
\label{table:comparison_cond_ind}
\end{table*}


For comparison of various approaches, see Table.~\ref{table:comparison_cond_ind}. The key differences are:
\begin{itemize}
    \item Classical: Point estimate (0.56) assumes precise knowledge and perfect factorization.
    \item Fr\'echet: Conservative interval [0.42, 0.80] guarantees containment but is wide (0.38 width).
    \item Strong: Narrower interval [0.42, 0.72] (0.30 width) but requires strong factorization assumption.
\end{itemize}
\noindent
Regarding practical implications, we have:
\begin{itemize}
    \item Diagnosis: If we need $>0.7$ probability to diagnose:
    \begin{itemize}
        \item Classical: No (0.56 < 0.7)
        \item Fr\'echet: Maybe (0.42-0.80 includes >0.7)
        \item Strong: Maybe (0.42-0.72 includes >0.7)
    \end{itemize}
    \item Safety: Fr\'echet is safer (always contains true probability).
    \item Efficiency: Strong independence is more efficient (narrower intervals).
\end{itemize}
\noindent
Finally, regarding when to use each method:
\begin{itemize}
    \item Fr\'echet Rule: Safety-critical applications, unknown dependencies, conservative risk assessment.
    \item Strong Independence: When independence of negative evidence is justified, efficiency is priority.
    \item Classical: When parameters are precisely known and independence assumptions are well-validated.
\end{itemize}
\end{example}

\section{Iterated Function Systems with Imprecise Probabilities}

In this section, we continue to bridge concepts from classical capacity theory \cite{Choquet1954,Walley1991} with domain theory by considering the domain that is dual to the event domain of a basic space. The dual domain takes pairs of covering closed subsets of the basic space with reverse inclusion. Using this dual domain, we formulate a new family of credal sets as the invariant measure of an iterated function system (IFS) with probabilities. 

Iterated Function Systems (IFS) and their invariant measures have been extensively studied 
in fractal geometry and dynamical systems, with applications ranging from computer graphics 
and image compression to natural phenomena modelling, signal processing, biological structure 
analysis, and financial time series \cite{fernau2022iterated}. 
Our results in this section provide a mathematically sound foundation for extending these classical 
applications to settings where the probabilities are uncertain or partially specified.

\begin{definition}
    For a basic space $D$, the {\em dual event domain} $\mathbb{E}_C(D)$ consists of pairs of closed subsets $(C_1,C_2) $ of $D$ with $C_1\cup C_2=D$ ordered by reverse inclusion. 
\end{definition}

We immediately obtain the following relation between the event domain and its dual.
Let $S^c$ denote the complement of the set $S$. 
\begin{proposition}
   The event domain $\mathbb{E}(D)$ is isomorphic with its dual  $\mathbb{E}_C(D)$ with the isomorphism given by the map \[(O_1,O_2)\mapsto (O_2^c,O_1^c):\mathbb{E}(D)\to \mathbb{E}_C(D)\]
\end{proposition}
\noindent
Given this isomorphism, we can deduce the next three results, namely Lemma~\ref{max-attaining}, Proposition~\ref{scott-cont-dual-domain}, and Corollary~\ref{upper-envelop}  using the densely injective properties of bounded complete domains as we did in the earlier sections of the paper. However, it is more instructive and more intuitive to derive them directly, as we shall do here.

Assume that the basic space $X$ is Hausdorff. Given $\nu\in P(D)$, we
obtain a mapping $\hat{\nu}: U(D)\to [0,1]$, where $U(D)$ is the upper
space of $D$, with $\hat{\nu}(C)=1-\nu(C^c)$. It is easy to check that
$\hat{\nu}$ is Scott continuous.

\begin{lemma}\label{max-attaining}
  For a given closed $C\subseteq X$, the evaluation map
  $\operatorname{ev}_C:P(X)\to \mathbb{R}$ with $\nu\mapsto \hat{\nu}(C)$
  attains its maximum on every compact subset of $U(P(D))$.
\end{lemma}

\begin{proof}
  For a closed set $C\subset X$, the evaluation map
  $\operatorname{ev}_C:P(X)\to \mathbb{R}$ with
  $\nu\mapsto \hat{\nu}(C)$ is an upper semi-continuous (usc) map with
  respect to the weak topology on $P(X)$. Since an usc map attains its
  supremum on the compact set $K$, it follows that there exists
  $\mu\in K$ such that $\mu(C)=\max\{\hat{\nu}(C):\nu\in K\}$.
\end{proof}

In general, the evaluation map $\operatorname{ev}_C$ does not attain
its infimum on a compact set as can be illustrated with examples for
$X=[0,1]$. We will also see an example of IFS with probability for
which the set $\{\hat{\nu}(C):\nu\in K\}$ is not an interval.

For a compact $X$, compact set $A\in U(P(X))$, and a given closed set
$C\subseteq X$, we denote by $\nu^A_C\in A$ any element with
$\nu^A_C(C)=\max \{\nu(C): \nu\in A\}$. Let $\mathcal{K}(X)$ denote
the collection of closed (and thus compact) subsets of $X$, ordered by
reverse inclusion.
\begin{proposition}\label{scott-cont-dual-domain}
    For any $A\in U(P(X))$,
    The mapping
    \begin{align*}
      \operatorname{\nu}^A_{(-)}(-): \mathcal{K}(X) &\to
                                                      ([0,1],\geq)\\
      C &\mapsto \operatorname{\nu}^A_C(C)
    \end{align*}
is Scott continuous. 
\end{proposition}

\begin{proof}
  Let $C\supseteq C'$; then
  $\nu^A_C(C) \geq \nu^A_{C'}(C)\geq \nu^A_{C'}{C'}$, thus the mapping
  is monotone.  Let $C_0\supseteq C_1\supseteq C_2 \supseteq\cdots $
  be a decreasing sequence of closed subsets of $X$ with
  $C=\bigcap_{i\in \nat}C_i$. We will show that
  $\inf_{i\in\nat} \nu^A_{C_i}(C_i)=\nu^A_{C}(C)$. Let $r\in (0,1)$
  satisfy $r>\nu^A_C(C)$. Then we have $r>\nu^A_C(C)\geq \nu(C)$ for
  any $\nu\in A$ by the definition of
  $\nu^A_C$. Thus, \begin{equation}\label{all-m} \forall \nu\in A.\,
    \exists i(\nu)\in \nat.\,r>\nu(C_{i(\nu)})\end{equation}Since the
  map $\operatorname{ev}_{C_i}:X\to [0,1]$ is usc, it follows that the
  set $U_i=\{\nu\in X: r>\nu(C_i)\}$ is weak open in the subspace weak
  topology on $X$. We also have $U_i\subset U_{i+1}$ since
  $C_i\supseteq C_{i+1}$. By Expression~(\ref{all-m}), it follows that
  $\bigcup_{i\in \nat}U_i=X$. By the compactness of $X$, there exists
  $i\in \nat$ such that $U_{i}=X$. It follows that for all $\nu\in A$
  we have $r>\nu^A_{C_i} (C_{i})$. Since $r>\nu^A_C(C)$ was arbitrary,
  it follows that $\mu^A_C(C)\geq \inf_{i\in\nat} \mu^{A}_{C_i}(C_i)$
  and Scott continuity follows.
\end{proof}

Finally, we conclude:
\begin{corollary}\label{upper-envelop}
  The mapping
  \begin{align*}
    {\mathcal L}: \Big(U(P(X)), \mathbb{E}_C(X)\Big) &\to \mathbb{I}[0,1]\\
    (A,(C_1,C_2)) &\mapsto
                    \left[1-\nu^A_{C_1}({C_1}),\nu^A_{C_2}(C_2)\right],
  \end{align*}
   where $\nu^A_{C_i}(C_i)=\max\{\nu(C_i):\nu\in A\}$, for $i=1,2$,
   is Scott continuous. 
\end{corollary}
Corollary~\ref{upper-envelop} thus provides a domain-theoretic realisation of the capacity-theoretic envelope: the map \(\mathcal{L}\) assigns to each credal set \(A\) and closed event pair \((C_1, C_2)\) the precise interval \[[1 - v_{C_1}^A(C_1), v_{C_2}^A(C_2)],\] where the upper probabilities \(v_{C_i}^A\) are defined via maximisation over \(A\). This construction captures the classical notion of upper and lower probabilities as envelopes of a credal set \cite{Walley1991,Augustin2014} within a domain-theoretic framework. The Scott continuity of \(\mathcal{L}\) ensures that this envelope can be approximated from finite or partial information---a crucial property for computational implementations of imprecise probabilistic reasoning. Moreover, by framing the result in terms of the dual event domain \(\mathbb{E}_C(X)\), we maintain a duality with the earlier open-set event domain \(\mathbb{E}(X)\), thereby unifying open-set and closed-set representations of uncertainty under a single domain-theoretic roof.

In what follows we consider credal sets of IFSs. The case of Markov chains having imprecise input given as intervals is presented in Appendix~\ref{appendix:interval_markov}.

\subsection{Credal sets of IFSs with probabilities}

In this subsection, we introduce a new family of credal sets, namely those consisting of invariant measures of iterated function systems (IFSs) with probabilities. The theory of IFSs has been an active area of research in multiple disciplines.

An iterated function system with probabilities is given by a finite
set of contracting maps $f_i:X\to X$ on a complete metric space $X$,
each associated with a probability weight $p_i$ with
$\sum_{i=1}^np_i=1$. We write $f$ for the vector function $(f_i)_i$
and $p$ for the probability vector $(p_i)_i$. Hutchinson showed that
such a system has a unique invariant measure which is the fixed point
of the Markov operator $H_{f,p}:P(X)\to P(X)$ on the space $P(X) $ of
probability measures on $X$, which is given by:
\[H_{f,p}(\mu)=\sum_{i=1}^n p_i\cdot \mu\circ f_i^{-1}.\]

The support of the invariant probability measure lies on the compact set $R$ that satisfies the fixed point equation $R=\bigcup_{i=1}^nf_i[R]$. Since the support of the invariant measure $\operatorname{fix}(H_{f,p})$ lies in the compact set $R$, we can assume without loss of generality that $X$ is itself a compact metric space.

Denote this unique fixed point as $\operatorname{fix}(H_{f,p})$. If the components of the probability vector are interval valued with $p\in \prod_{1\leq i\leq n} [p_i^-,p_i^+]$, where $n>1$, the admissible set, assumed to be non-empty, is defined as:
\begin{equation}
\label{eq:A_admissible_p_i}
K=S^{n-1}\cap \prod_{1\leq i\leq n} [p_i^-,p_i^+],
\end{equation}
where $S^{n-1}=\{p\in [0,1]^n: \sum_{i=1}^np_i=1\}$ is the $n-1$-dimensional simplex. Since $S^{n-1}$ and $\prod_{1\leq i\leq n} [p_i^-,p_i^+]$ are convex polytopes, so is their intersection $K$.

\begin{theorem}~\cite{centore1994continuity}\label{IFS-fix}
For a fixed vector function $f$,  the map $\operatorname{fix}(H_{f,-}):K\to P(X)$ with $p\mapsto \operatorname{fix}(H_{f,p})$ is continuous with respect to the Euclidean subspace topology on $S^{n-1}$ and the weak topology on $P(X)$. 
    
\end{theorem}
\begin{lemma}\label{shrinking}
    If $g:X\to Y$ is a continuous map of Hausdorff spaces $X$ and $Y$ then for any shrinking sequences of compact non-empty sets $K_i$, $i\in \nat$, we have $g[\bigcap_{i\in \nat} K_i]=\bigcap_{i\in \nat} g[K_i]$.
\end{lemma}
\begin{corollary}
  The map
  \begin{align*}
    \operatorname{fix}(H_{f,[-]}): U(K) &\to U(P(X))\\
    A &\mapsto\operatorname{fix}(H_{f,[A]})=\{ \operatorname{fix}(H_{f,p}):p\in A\}
  \end{align*}
  is well-defined and Scott continuous.
\end{corollary}

\begin{proof}
   Since the continuous image of a compact set is compact, it follows by Theorem~\ref{IFS-fix} that the set \[\{ \operatorname{fix}(H_{f,p}):p\in A\}\subseteq P(X)\] is compact and hence in $U(P(X))$. By Lemma~\ref{shrinking}, we conclude that the map $\operatorname{fix}(H_{f,[-]})$ is Scott continuous.
\end{proof}

The continuous probability valuation corresponding to the invariant measure can be obtained domain-theoretically as the least fixed point of the Hutchinson operator on the probabilistic power domain of the upper space of $X$ if $X$ is compact or, the least fixed point of the Hutchinson operator on the probabilistic power domain of the domain of formal balls of the complete metric space. 

And we obtain an extension of the fixed-point operator to the admissible set:
\[\operatorname{fix}(A)=\{\operatorname{fix}(H_{f,p}):p\in A\} \subseteq P(X).\]

 We note that the image of $\operatorname{ev}_C:K\to [0,1]$ may not be an interval. For example, consider the IFS with two maps on $[0,1]$ that induces the classical Cantor set: $f_1:x\mapsto x/3$ and $f_2:x\mapsto 2/3+x/3$. Let $p_1\in[0,r]$ and $p_2\in [1-r,1]$ with $0<r\leq 1/2$. Then it is easy to see that the invariant measure corresponding to $p_1=0$ and $p_2=1$ is the Dirac measure $\delta$, whereas for $p_1=r$ and $p_2=1-r$, we obtain a non-atomic invariant measure $\nu$. Thus, for the closed set $C:=\{1\}$, the image of   $\operatorname{ev}_{\{1\}}:K\to [0,1]$ will be $\{0,1\}$.

Any subset $V_0$ of the vertices $V$ of $A$ gives, for any closed set $C\subset X$, a set of points:
\begin{equation*}
\{\operatorname{fix} (H_{f,p})(C):p\in V_0\}\subseteq\{\operatorname{fix} (H_{f,p})(C):p\in A\}\subseteq [0,\mu^A_C(C)].    
\end{equation*}

\noindent
Using the general result in Corollary~\ref{upper-envelop}, we can now deduce:

\begin{corollary}
  The composition
  \begin{align*}
    \mathcal{L}\circ \operatorname{fix}(H_{f,[-]}): U(K) &\to U(P(X))\to
    \mathbb{E}_C(X)\to \mathbb{I}[0,1] \\
    A &\mapsto \mathcal{L}(\operatorname{fix}H_{f,[A]})
  \end{align*}
  is Scott continuous.
\end{corollary}

Thus, given an event represented by the covering closed sets $(C_1,C_2)$, the  map $\mathcal{L}\circ \operatorname{fix}(H_{f,[-]})$ gives an outer (envelope) approximation of the set of continuous valuations $\{\operatorname{fix}(H_{f,p}): p\in A\} $ from below by $1-\mu_{C_1}^A(C_1)$ and from above by $\mu_{C_2}^A(C_2)$. 

\subsubsection{Markov chains with imprecise transition matrix} In this section, we have introduced and analysed iterated function systems with imprecise probability weights. 
Analogously, Appendix~\ref{appendix:interval_markov} provides, to our knowledge, a novel treatment of 
finite-state Markov chains whose transition probabilities are specified with imprecise values.

\section*{Conclusion}

We have developed a comprehensive domain-theoretic foundation for imprecise probability and credal sets, providing a unified computational framework that handles both partial event descriptions (represented as pairs of disjoint opens in the event domain $\mathbb{E}(D)$) and distributional uncertainty captured by compact credal sets in the upper space $U(P(X))$. Our main contributions include: the construction of Scott-continuous interval probability maps on event domains; a monotonicity-derived, sharp interval extension of Bayesian updating with accompanying sound and complete inference rules; a theory of conditional independence for imprecise events with both conservative (Fr{\'e}chet) and strong factorisation rules; and the introduction of a novel family of credal sets generated by iterated function systems with imprecise probability weights, with a proof of continuity for the associated fixed-point map. We have formulated logical predicates for the relevant concepts, and derived soundness and completeness results. All operations extend Scott-continuously to spaces of credal sets, ensuring that finite approximations converge.

This work establishes the necessary mathematical infrastructure for a domain-theoretic treatment of credal networks and imprecise Bayesian networks. The framework guarantees that inference with imprecise parameters and partially specified observations is computationally grounded, supported by approximation properties from domain theory and a logical foundation via approximable mappings. Future work will focus on constructing explicit domain-theoretic credal networks, developing exact and approximate inference algorithms that leverage the continuity and approximation structure presented here, and extending the approach to sequential decision making under imprecision.



\newcommand{\etalchar}[1]{$^{#1}$}

\appendix

\section{A continuous valuation that is not Lawson continuous}

\label{appendix:not_Lawson_cont}

 Recall that the Lawson topology on a domain is the join of the Scott topology and the lower topology with a sub-basis of open sets given by $D\setminus \uparrow x$. It follows that a Scott continuous function $f:A\to B$ on domains $A$ and $B$ is Lawson continuous if and only if it preserves the infimum of any filtered family of Scott open sets. 

In the example below, we construct a continuous valuation on the unit interval which is not Lawson continuous.

\begin{example}Let $0<r<1$ and consider the construction of a fat Cantor set with support on a fractal set of Lebesgue measure $1-r$ as follows. The scheme of constructing the standard Cantor set in $[0,1]$ is mimicked, but the length of each removed open interval is scaled by $r$. Thus, in the first iteration the middle open interval of length $r/3$ is removed, in the second iteration two middle open intervals each of length $r/3^2$ are removed and in the $n$th iteration we remove $2^n$ middle intervals each of length $r/3^{n+1}$.The resulting Cantor set $C$ has Lebesgue measure $1-r>0$. If $s$ is the continuous valuation corresponding to this fractal measure on $[0,1]$, then it is easy to see that it does not preserve infimum of filtered collection of open sets. Let $O_n$ be the union of the interiors of the closed intervals that remain at iteration $n$ after removing the $2^n$ middle open intervals. Then $s(O_n)=1-r  + r(2/3)^{n}$ for $n\in \nat$ and the family $(O_n)_{n\in \nat}$ is filtered. However $\bigwedge_{n\in \nat} O_n=\emptyset$ since $\bigcap_{n\in \nat} O_n=C\setminus B$, where $B$
is the countable set of the end points of the removed intervals at all rounds of iteration. Thus, \[\inf_{n\in \nat} s(O_n)=1-r\neq 0=s((\emptyset)=s\left(\bigwedge_{n\in \nat}O_n\right).\]
\end{example}

 \section{Interval Based Markov Chains}

\label{appendix:interval_markov}
 
Consider a finite discrete space of $n$ elements indexed by integers $1\leq i\leq n$. We assume that the Markov transition matrix $T=(t_{ki})$ has imprecise entries given as intervals $[p_{ki},q_{ki}]$, with $1\leq k,i\leq n$ and that any admissible probability matrix satisfying these conditions is aperiodic and irreducible. In the literature on statistics with imprecise input, the problem is tackled with iteration and linear programming. We propose a simpler technique, based on finding the minimum and maximum of each component of the fixed point of the transition matrix over all the vertices of the polytope of admissible matrices. The method in general requires bilinear programming, but it gives a complete answer for $n=2$ and an inner approximate to the solution set which is feasible to compute, for small $n$. 

We first observe that the admissible set $A$ of transition matrices, 
\begin{equation}
\label{eq:A_admissible_t_k_i}
A:=\left\{t_{ki}:p_{ki}\leq t_{ki}\leq q_{ki}, \sum_{i=1}^nt_{ki}=1,1\leq k,i\leq n\right\}\subseteq\mathbb{R}^{n\times n}.     
\end{equation}
which we assume is non-empty, is a convex set since the intersection of convex sets is convex and the $n$-dimensional simplex $ S^{n-1}:=\{t_{ki}:\sum_{i=1}^nt_{ki}=1\}$ for $1\leq k\leq n$ is convex, as is each semi-space 
\[\{t_{ki}: p_{ki}\leq t_{ki} \} \qquad \mbox{or} \qquad \{t_{ki}: t_{ki}\leq q_{ki} \}\]
for $1\leq k,i\leq n$.
In addition, $A$ is compact since $S^{n-1}$ is compact and the semi-spaces are closed. Let $T^j$ for $j\in J$ be the finite set of vertices of $A$. Let $\pi^T$ for $T\in A$ be the left fixed point of $T$:
\[\pi^TT=\pi^T\]
Then, for $1\leq k\leq n$, we have 
\[m_k=\min \{(\pi^T)_k: T\in A\}\leq \min\{(\pi^{T^j})_k:j\in J\}:=m^J_k\]
\[M_k^J:=\max\{(\pi^{T^j})_k:j\in J\}\leq \max \{(\pi^T)_k: T\in A\}:=M_k\]
Since each matrix in $A$ is assumed to be aperiodic and irreducible, the mapping $\pi_k:A\to \mathbb{R}$ with $\pi_k(T)=\pi_k^T$ is continuous for $1\leq k\leq n$. It follows that $\pi_k$ maps the compact and convex set $A$ 
onto a compact none-empty interval of $\mathbb{R}$. Thus, for any subset $I\subset J$, we have the approximating interval 
\[\left[m^I_k,M^I_k\right]\subseteq \left[m^J_k,M^J_k\right]\subseteq [m_k,M_k]\]
For $n=2$, we can actually obtain the interval $[m_k,M_k]$ easily. We can write the matrix $T$ and its left fixed point as $(x,1-x)$ and
\[T=\begin{bmatrix}
    t_{11}&1-t_{11}\\
    t_{21}&1-t_{21}\\
\end{bmatrix}\]
This gives $x=t_{21}/(1-t_{11}+t_{21})$ and we have both $\partial x/\partial t_{11},\partial x/\partial t_{21}>0$. Hence, the minimum is attained for $t_{11}=p_{11}$ and $t_{21}=p_{21}$, while the maximum is attained for $t_{11}=q_{11}$ and $t_{21}=q_{21}$.

Here are the ingredients of the bilinear programming problem for a general solution to the problem: Minimise $\pi_k$ subject to
\[\begin{array}{llr}
\pi_i - \sum_{k=1}^n \pi_k t_{ki} = 0 & \text{for } i = 1, \dots, n &\text{(Flow/Invariance)} \\
 \sum_{k=1}^n \pi_k = 1& & \text{(Normalisation)} \\
 p_{ki} \le t_{ki} \le q_{ki} & \text{for all } k, i \in \{1, \dots, n\} & \text{(Input Bounds)} \\
 \sum_{i=1}^n t_{ki} = 1 & \text{for } k = 1, \dots, n &\text{(Row Stochasticity)} \\
 \pi_i \ge 0 & \text{for } i = 1, \dots, n & \text{(Non-negativity)}
\end{array}
\]

\end{document}